\documentclass[a4paper,12pt]{article}
\usepackage{authblk}
\title{Identifying a minimal class of models for high-dimensional data}
\author[1,2]{Daniel Nevo}
\author[1,3]{Ya'acov Ritov\thanks{Supported by part by ISF grant 1770/15.} }
\affil[1]{Department of Statistics, The Hebrew University of Jerusalem}
\affil[2]{Departments of Biostatistics and Epidemiology\\

 Harvard T.H. Chan School of Public Health}
 
\affil[3]{Department of Statistics, University of Michigan}

\date{}
\usepackage{jrSmall}
\usepackage{apajr,comment}

\usepackage{setspace}
\usepackage{bbm}
\usepackage[font=footnotesize,labelfont=bf]{caption}
\usepackage{color}
\usepackage{enumerate}
\usepackage[shortlabels]{enumitem}
\usepackage{floatpag}
\usepackage[margin=1.25in]{geometry}
\usepackage{graphicx}
\usepackage{graphics}
\usepackage{longtable}
\usepackage{multirow}
\usepackage{subcaption}
\usepackage{setspace}
\usepackage{rotating}

\newcommand{\bL}{\hat{\beta}^L}
\newcommand{\bEN}{\hat{\beta}^{EN}}
\newcommand{\bz}{\beta^0}

\makeatletter
\newcommand*{\rom}[1]{\expandafter\@slowromancap\romannumeral #1@}
\makeatother

\begin{document}

\maketitle
\begin{abstract}
Model selection consistency in the high-dimensional regression setting can be achieved only if strong assumptions are fulfilled. We therefore suggest to pursue a different goal, which we call a minimal class of models. The minimal class of models includes models that are similar in their prediction accuracy but not necessarily in their elements. We suggest a random search algorithm to reveal candidate models. The algorithm implements simulated annealing while using a score for each predictor that we suggest to derive using a combination of the Lasso and the Elastic Net. The utility of using a minimal class of  models is demonstrated in the analysis of two datasets.
\end{abstract}
\section{Introduction}
\label{Sec:intro}

High dimensional statistical problems have been arising as a result of the vast amount of data gathered today. A more specific problem is that estimation of the usual linear regression coefficients vector cannot be performed when the number of predictors exceeds the number of observations. Therefore,  a sparsity assumption is often added. For example, the number of regression coefficients that are not equal to zero is assumed to be small.
If it was known  in advance which predictors have non zero coefficients, the classical linear regression estimator could have been used. Unfortunately, it is not known. Even worse, the natural relevant discrete optimization  problem is usually not computationally feasible.

The Lasso estimator, \citeshort{tibshirani1996regression}, which solves the problem of minimizing prediction error together with a $\ell_1-$norm penalty, is possibly the most popular method to address this problem, since it results in a sparse estimator. Various algorithms are available to compute this estimator [e.g., \citeshort{friedman2010regularization}].  The theoretical properties of the Lasso  have been throughly researched  in the last decade. For the high-dimension problem, prediction rates were established in various manners, \citeshort{greenshtein2004persistence,bunea2006aggregation,bickel2009simultaneous,bunea2007sparsity,meinshausen2009lasso}. The capability of the Lasso to choose the correct model depends on the true coefficient vector and the matrix of the predictors,  or more precisely, on its Gram matrix, \citeshort{meinshausen2006high,zhao2006model,zhang2008sparsity}. However, the underlying assumptions are typically rather restrictive, and cannot be checked in practice.

In order to overcome its initial disadvantages, many modifications of the Lasso were suggested. For example,  the Adaptive Lasso, \citeshort{zou2006adaptive}, is a two stage procedure with a second step weighted Lasso, that is, some predictors get less penalty than others;
When a grouped structure of the predictors is assumed, the Group Lasso, \citeshort{yuan2006model}, is often used; The Elastic Net estimator, \citeshort{zou2005regularization}, is intended to deal with correlated predictors. It is obtained by adding a penalty on the $\ell_2-$norm of the coefficients vector together with the $\ell_1$ Lasso penalty. \citeshort{zou2005regularization} also empirically found that the Elastic Net's prediction accuracy is better than the Lasso's.

In the high-dimensional setting, the task of finding the true model might be too ambitious, if meaningful at all. Only in certain situations, which could not be identified in practice, model selection consistency is guaranteed. Even in the classical setup, with more observations than predictors, there is no model selection consistent estimator unless further assumptions are fulfilled. This leads us to present a different objective. Instead of searching for a single ``true'' model, we aim to present a number of possible models a researcher should look at. Our goal, therefore, is to find potential good prediction models. Since data are not generated by computer following one's model, there is a benefit in finding several models with similar performance if they exist. In short, we suggest to find the best models for each small model size. Then, by looking at these models one may reach interesting conclusions regrading the underlying problem. Some of these, as we do below, can be concluded using statistical reasoning, but most of these should be reasoned by a subject matter expert.

In order to find these models, we implement a search algorithm that uses simulated annealing, \citeshort{kirkpatrick1983optimization}. The algorithm is provided with a ``score'' for each predictor that we suggest to get using a multi-step procedure that implements both the Lasso and the Elastic Net (and then the Lasso again). Multi-step procedures in the high-dimensional setting have drawn some attention and were demonstrated to be better than using solely the Lasso, \citeshort{zou2006adaptive,bickel2010hierarchical}.

The rest of the paper is organized as follows. Section \ref{Sec:descrip} presents the concept of minimal class of models and the notations.
 Section \ref{Sec:algo} describes a search algorithm for relevant models, and gives motivation for the sequential use of the Lasso and the Elastic Net when calculation a score to each predictor. Section \ref{Sec:numerRes} consists of a simulation study and two examples of data analysis using a minimal class of models.  Section \ref{Sec:disc} suggests a short discussion. Technical proofs and supplementary data are provided in the appendix.
\section{Description of the problem}
\label{Sec:descrip}
We start with notations.  First, denote $||v||_q:=(\sum v_j^q)^{1/q}, \; q>0$ for the $\ell_q$ (\emph{pseudo}) norm of any vector $v$, $\|v\|_0 =\lim_{q\to 0}||v||_q$, the cardinality of $v$.
The data consist of  a  predictors matrix, $X_{n \times p}=(X^{(1)}\  X^{(2)}\ ...\ X^{(p)})$ and  a response vector, $Y_{n\times 1}$. WLOG, $X$ is centered and scaled
 and $Y$ is centered as well.  We are mainly interested in the case $p>n$. The underling model is $Y=X\beta+\epsilon$ where $\epsilon_{n\times 1}$ is a random error, $E(\epsilon)=0, \; V(\epsilon)=\sigma^2I$, $I$ is the identity matrix. $\beta$ is an unknown parameter and its true value is denoted by $\bz$.

 Denote $S\subseteq \{1,...,p\}$ for a set of indices of $X$. We call $S$  \textit{a model}. We use $s=|S|$ to denote the cardinality of the set $S$. Denote also $S_0:=\{j:\bz\ne 0\}$ and $s_0=|S_0|$ for the true model, and its size, respectively. For any model $S$, we define $X_S$ to be the submatrix of $X$ which includes only the columns specified by $S$.  Let $\hat{\beta}^{LS}_S$ to be the usual least square (LS) estimator corresponding to a model $S$, that is,
\begin{equation*}
\hat{\beta}^{LS}_S= (X_S^T X_S)^{-1}X^T_SY,
\end{equation*}
provided $X_S^T X_S$ is non singular.

Now, the straightforward approach to estimate $S_0$ given a model size $\kappa$ is to consider the following optimization problem:
\begin{equation}
\label{Eq:ellzerooptim}
\underset{\beta}{\text{min}}  \frac{1}{n}||Y-X\beta||_2^2,  \qquad \text{s.t } \quad ||\beta||_0 = \kappa.
\end{equation}
Unfortunately, typically,  solving \eqref{Eq:ellzerooptim}  is  computationally infeasible. Therefore, other methods were developed and are commonly used. These methods produce sparse estimators and can be implemented relatively fast. We first present here the Lasso,  \citeshort{tibshirani1996regression}, defined as
\begin{equation}
\label{Eq:LasDef}
\bL = \underset{\beta}{\text{argmin}} \Bigl(\frac{1}{n}||Y-X\beta||_2^2+\lambda||\beta||_1\Bigr)
\end{equation}
where $\lambda>0$ is a tuning constant.  For some applications, a different amount of regularization is applied for each predictor. This is done using the weighted Lasso, defined by
\begin{equation}
\label{Eq:WeiLasDef}
\bL_w = \underset{\beta}{\text{argmin}}\Bigl( \frac{1}{n}||Y-X\beta||_2^2+\lambda||w\cdot \beta||_1\Bigr)
\end{equation}
where $w$ is a vector of $p$ weights, $w_j\ge 0$ for all $j$, and $a\cdot b$ is the Hadamard (Schur, entrywise) product of two vectors $a$ and $b$. The Adaptive Lasso, \citeshort{zou2006adaptive}, is one example of using a weighted Lasso type estimator.
Next is the Elastic Net estimator
\begin{equation}
\label{Eq:ENDef}
\bEN = \underset{\beta}{\text{argmin}} \Bigl( \frac{1}{n}||Y-X\beta||_2^2+\lambda_1||\beta||_1+\lambda_2||\beta||^2_2\Bigr).
\end{equation}
This estimator is often described as a compromise between the Lasso and the well known Ridge regression, \citeshort{hoerl1970ridge}, since it could be rewritten as
\begin{equation}
\label{Eq:ENDefAlp}
\bEN = \underset{\beta}{\text{argmin}} \Bigl( \frac{1}{n}||Y-X\beta||_2^2+\lambda\bigl(\alpha||\beta||_1+(1-\alpha)||\beta||^2_2\bigr)\Bigr).
\end{equation}

Let $\hat{\beta}_n$  be a sequence of estimators for $\beta$ and let $\hat{S}_n$ be the sequence of corresponding models.  Model selection consistency is commonly defined as
\begin{equation}
\label{Eq:mod.sel.cons}
\lim\limits_{n\rightarrow \infty} P(\hat{S}_n=S_0)=1.
\end{equation}
If $p\ll n$ and small, then criteria based methods e.g., BIC, \cite{schwarz1978estimating}, can be model selection consistent if $p$ is fixed or if suitable conditions are fulfilled, c.f.,  \citeshort{wang2009shrinkage} and references therein. However, these methods are rarely computationally feasible for large $p$. For $p>n$, it turns out that practically  strong and unverifiable  conditions are needed to achieve \eqref{Eq:mod.sel.cons} for popular regularization based  estimators such as the Lasso: \cite{zhao2006model} and \cite{meinshausen2006high}; the Adaptive Lasso: \citeshort{huang2008adaptive}; the Elastic Net: \cite{jia2010model}; but also for Orthogonal Matching Pursuit (OMP), which is essentially forward selection: \cite{tropp2004greed} and \cite{zhang2009consistency}.

In light of these established results,
we suggest to pursue a different goal. Instead of finding a single model, we suggest to look for a group of models. Each of these models should include low number of predictors, but it should also be capable of  predicting $Y$ well enough.    Finally, $\mathcal{G}=\mathcal{G}(\kappa,\eta)$ is called a minimal class of models of size $\kappa$ and efficiency $\eta$ if
\begin{equation}
\label{Eq:Gdef}
\mathcal{G}=\Bigl\{S: |S|=\kappa \;\&\; \frac{1}{n}||Y-X_S\hat{\beta}_S^{LS} ||_2^2\le \underset{|S'|=\kappa}{\min} \{\frac{1}{n}||Y-X_{S'}\hat{\beta}_{S'}^{LS} ||_2^2\}+\eta\Bigr\}.
\end{equation}
 One could control how similar the models in $\mathcal{G}$ are to each other in terms of prediction, using the tuning parameter $\eta$. A reasonable choice is $\eta=c\sigma^2$ with some $c>0$.  If $\sigma^2$ is unknown, it could be replaced with an estimate, e.g., using the Scaled Lasso, \cite{sun2012scaled}. An alternative to $\mathcal{G}$ is to generate the set of models by simply choosing for each $\kappa$ the $M$ models having  the smallest sample MSE, for some number $M$. The LS estimator, $\hat{\beta}_S^{LS}$,  minimizes the sample prediction error for any model $S$ with size $s \le n$ $\hat{\beta}_S^{LS}$. Thus, this estimator is  used for each of the considered models.

 Note that $\mathcal{G}$ depends on $\kappa$, the desired model size. However, in practice one may want to find $\mathcal{G}$ for a few values of $\kappa$, e.g., $\kappa=1,...,10$, and then to examine the pooled results, $\bigcup_{j=1}^k \mathcal{G}(j,\eta)$. Another option is to replace the Mean Square Error (MSE) $n^{-1}||Y-X_S\hat{\beta}_S^{LS} ||_2^2$ in the definition of $\mathcal{G}$  with one of the available model selection criteria, e.g., AIC ,\citeshort{akaike1974new}, BIC or Lasso. Note that we are interested in situations where there are fair models with a relatively very small number of explanatory variables out of the $p$ available.

At this point, a natural question is how can we benefit from using a minimal class of models.  Examining the models in $\mathcal{G}$ may allow us to derive conclusions regarding the importance of different explanatory variables. If, for example, a variable appears in all the models that belong to $\mathcal{G}$, we may infer that it is essential for prediction of $Y$, and cannot be replaced. We demonstrate this kind of analysis in Section \ref{SubSec:data}.

Another possibility is to use one out of the many aggregation of models methods estimates. Aggregation of estimates obtained by different models was suggested both for the frequentist, \citeshort{hjort2003frequentist}, and for the Bayesian, \citeshort{hoeting1999bayesian}. The well known ``Bagging'', \citeshort{breiman1996bagging}, is also a technique to combine results from various models. Averaging across estimates obtained by multiple models is usually carried out to account for the uncertainty in the model selection process. We, however, are not interested in improving prediction per se, but in identifying good models. Nor are we  interested in identifying the best model, since this is not possible in our setup, but in identifying variables that are potentially relevant and important.

\subsection{Relation to other work}
\label{SubSec:Relation}
 A similar point of view on the relevance of a variable was given by \cite{bickel2012discussion}. They considered a variable to be important  if its relative contribution to the predictive power of a set of variables is high enough. Their next step was to consider only specific type of sets, such that their prediction error is high, yet they do not contain too many variables.

\cite{rigollet2012sparse} investigated the relevant question of prediction under minimal conditions. They showed that linear aggregation of  estimators is beneficial for high-dimensional regression when assuming sparsity of the number of estimators included in the aggregation. They also showed that choosing exponential weights for the aggregation corresponds to minimizing a specific, yet relevant, penalized problem. Their estimator, however, is computationally impossible and they have little interest in variables and model identification.

As described in Section \ref{Sec:algo}, our suggested search algorithm for candidate models travels through the model space. We choose to use simulated annealing to prevent the algorithm from getting stuck in a local minimum.  Various Bayesian model selection procedures consists moves along the model space, usually using a relevant posterior distribution, cf. \cite{o2009review}. We, however, do not assume any prior distribution for the coefficient values. Our use of the algorithm is only as a search mechanism, simply to find as many as possible models that belong to $\mathcal{G}$. Convergence properties of the classical simulated annealing algorithm are not of interest to our use of it. We are interested in the path generated by the algorithm and not in its final state.

\section{A search algorithm}
\label{Sec:algo}
\subsection{Simulated annealing algorithm}
\label{SubSec:annealAlgo}
In this section, we suggest an algorithm to find $\mathcal{G}$ for a given $\kappa$ and $\eta$.  The problem is that $||Y-X_S\beta_S^{LS}||_2^2$ is unknown for all $S$, and since $p$ is large, even for a relatively small $\kappa$, the number of possible models is huge (e.g., for $p=200, k=4$ there are almost 65 million possible models). We therefore suggest to focus our attention on smaller set of models,  denoted by $\mathcal{M}(\kappa)$. $\mathcal{M}$ is a large set of models,  but not too large so we can calculate MSEs for all the models within $\mathcal{M}$ in a reasonable computer running time. Once we have $\mathcal{M}$ and the corresponding MSEs, we can form $\mathcal{G}$ by choosing the relevant models out of $\mathcal{M}$.

The remaining question is how to assemble $\mathcal{M}$ for a given $\kappa$. Any greedy algorithm is bound to find models that are all very similar. Our purpose is to find models that are similar in their predictive power, but heterogeneous in their structure.

Our approach therefore is to implement a search algorithm which travels between potentially attractive models. We use a simulated annealing algorithm. The simulated annealing algorithm was suggested for function optimization by \citeshort{kirkpatrick1983optimization}. The maximizer of a function $f(\theta)$ is of interest. Let $T=(t_1,t_2,...,t_R)$ be a decreasing set of positive ``temperatures''. For every temperature level $t\in T$, iterative steps are carried out, before moving to the next, lower, temperature level. In each step, a random suggested move from the current $\theta$ to another $\theta'\ne \theta$ is generated. The move is then accepted with a probability that depends on the ratio ${\exp\bigl[\bigl(f(\theta')-f(\theta)\bigr)/t\bigr]}$. Typically, although not necessarily, a Metroplis-Hastings criterion, \citeshort{metropolis1953equation,hastings1970monte}, is used to decide whether to accept the suggested move $\theta'$ or to stay at $\theta$.  Then, after a predetermined number of iterations $N_t$, we move to the next $t'<t$ in $T$, taking the final state in temperature $t$ as the initial state for $t'$. The motivation for using this algorithm is that for high ``temperatures'', moves that do not improve the target function are possible, so the algorithm does not get stuck in a small area of the parameter space. However, as we lower the temperature, the decision to move to a suggested point is based almost solely on the criterion of improvement in the target function value. The name of the algorithm and its motivation come from annealing in metallurgy (or glass processing), where a strained piece of metal is heated, so that a reorganization of its atoms is possible, and then it colds off so the atoms can settle down in low energy position. See \cite{brooks1995optimization} for a general review of simulated annealing in the context of statistical problems.

In our case, the parameter of interest is $\beta$, or more precisely, the model $S$. The objective function, that we wish to maximize, is
$$
f(S)=-\frac{1}{n}||Y-X_S\hat{\beta}_S^{LS}||^2_2.
$$
We now describe the proposed algorithm in more detail. We use simulated annealing with Metropolis-Hastings acceptance criterion as a search mechanism for good models. That is, we are not looking for the settling point of the algorithm, but we follow its path, hope that much of it will be in neighborhood of good models, and find the best models along the path.

We say the algorithm is in step $(t,i)$ if the current temperature is $t\in T$ and the current iteration in this temperature is $i\in \{1,...,N_t\}$. For simplicity, we describe here the algorithm for $N_t=N$ for all $t$.
Let $S_t^i$ and $\hat{\beta}^i_t$ be the model and the corresponding least square estimator in the beginning of the state $(t,i)$, respectively. An iteration includes a suggested model $S^{i+}_t$, a least square estimator for this model, $\hat{\beta}_t^{i+}$, and a decision whether to move to $S^{i+}_t$ and  $\hat{\beta}_t^{i+}$ or to stay at $S_t^i$ and $\hat{\beta}^i_t$. We now need to define how $S^{i+}_t$ is suggested and what is the probability of accepting this move.

For each $S^i_t$, we suggest $S^{i+}_t$ by a minor change, i.e., we take one variable out and we add another in, and then obtain $\hat{\beta}_t^{i+}$ by standard linear regression.
Assume that for every variable $j \in (1,...,p)$ we have a score $\gamma_j$, such that higher value of $\gamma_j$ reflects that the variable $j$ should be included in a model, comparing with other possible variables. WLOG, assume $0\le \gamma_j\le 1$ for all $j$.
We choose a variable $r^*\in S_t^i$ and take it out with the probability function
\begin{equation}
\label{Eq:ProbOut}
p^{out}_{i,r}=\frac{{\gamma_r}^{-1}}{\sum\limits_{u\in S_t^i}{\gamma_u}^{-1}}, \hspace{5 mm} \forall r\in S_t^i.
\end{equation}
Next, we choose a variable $\ell^* \notin S^i_t$ and add it to the model with the probability function
\begin{equation}
\label{Eq:ProbIn}
p^{in}_{i,\ell}=\frac{\gamma_\ell}{\sum\limits_{u\notin S_t^i}\gamma_u}, \hspace{5 mm} \forall \ell\notin S_t^i.
\end{equation}
Thus,
\begin{equation*}
S^{i+}_t = \{S^{i}_t\setminus r^*\} \cup  \{\ell^*\}
\end{equation*}
and we may calculate the LS solution $\hat{\beta}_t^{i+}$ for the model $S^{i+}_t$.
The first part of our iteration is over. A potential candidate was chosen. The second part is the decision whether to move to the new point or to stay at the current point.  Following the scheme of simulated annealing algorithm with Metropolis-Hastings criterion we calculate
 \begin{equation*}
 q = \exp\Bigl(\frac1{nt}\bigl(||Y-X_{S_t^{i}}\hat{\beta}_t^{i}||^2_2-||Y-X_{S_t^{i+}} \hat{\beta}_t^{i+}||^2_2\bigr)\Bigr)\frac{p(S^{i+}_t\rightarrow S^i_t)}{p(S^i_t\rightarrow S^{i+}_t)}
 \end{equation*}
where
\begin{equation}
\label{Eq:prob.model.inv}
\begin{split}
p(S^i_t\rightarrow S^{i+}_t)&=p^{out}_{i,r^*}\; p^{in}_{i,\ell^*}
\\
p(S^{i+}_t\rightarrow S^i_t)&=p^{out}_{i^+,\ell^*}\; p^{in}_{i^+,r^*}.
\end{split}
\end{equation}
We are now ready to the next iteration $i+1$ by setting
\begin{equation*}
(S_t^{i+1},\hat{\beta}_t^{i+}) = \left\{
		\begin{array}{ll}
		(S_t^{i+} ,\hat{\beta}_t^{i+})& \quad w.p \quad \min(1,q)\\
		(S_t^i ,\hat{\beta}_t^i) & \quad w.p \quad \max(0,1-q).\\
		\end{array}
\right.
\end{equation*}
Along the run of the algorithm, the suggested models and their corresponding MSEs are kept. These models are used to form $\mathcal{M}(\kappa)$, and $\mathcal{G}$ can be then identified for a given value of $\eta$.

We point out now several issues that should be considered when using the algorithm. First, the algorithm was described above for one single value of $\kappa$. In practice, one may run the algorithm separately for different values of $\kappa$. Another consideration is the tuning parameters of the algorithm that are provided by the user: The temperatures $T$; the number of iterations $N$; the starting point $S^{1}_{t_1}$; and the vector $\gamma=(\gamma_1,...,\gamma_p)$. Our empirical experience is that the first three can be managed without too many concerns; see Section \ref{Sec:numerRes}. Regarding the vector $\gamma$, a wise choice of this vector should improve the chance of the algorithm to move in desired directions. We deal with this question in Section \ref{SubSec:gamma}. However, in what follows we show that, under suitable conditions, the algorithm can work well even with a general choice of $\gamma$.

Define $S_0, s_0$ and $\bz$ as before and let $\mu=X\bz$. That is,  $Y=\mu+\epsilon$.
We first introduce a few simple and common assumptions:
\begin{Assumptions}
\item $ ||\mu||_2^2=\O(n)$
\assumption $s_0$ is small, i.e., $s_0=\O(1)$.
\assumption $p=n^a, a>1$
\assumption $\epsilon \sim N_n(0,\sigma^2I)$
\end{Assumptions}
Denote $A_\gamma$ for the set of positive entries in $\gamma$. That is, $A_\gamma \subseteq \{1,.2,...p\} $ is a (potentially) smaller group of predictors than all the $p$ variables. Denote  also $h_\gamma=|A_\gamma|$ for the size of $A_\gamma$ and $\gamma_{min}:=\min\limits_{i \in A_\gamma}\gamma_i$ for the lowest positive entry in $\gamma$.

Informally, the algorithm is expected to preform reasonably well if:
\begin{enumerate}[1.]
\item The true model is relatively small (e.g., with 10 active variables).
\item  A variable in the true model is adding to the prediction of a set of variables if a very few (e.g., 2) other variables are in the set.
\end{enumerate}

Our next assumption is more restrictive. Let $\bar{S}$ be an interesting model with size $s_0$---a model with not too many predictors  and with a low MSE. The models we are looking for are of this nature. We facilitate the idea of $\bar{S}$ being an interesting model by assuming that $X_{\bar{S}}\hat{\beta}_{\bar{S}}$ is close to $\mu$ (in the asymptotic sense). We virtually assume that for every model with $s=\bar{s}=|\bar{S}|$, which is not $\bar{S}$, if we take out a predictor that is not part of $\bar{S}$, and replace it with a predictor from $\bar{S}$, the subspace spanned by the new model is not much further from $\mu$, comparing with the subspace spanned by the original model. Formally, denote $\mathcal{P_S}$ for the projection matrix onto the subspace spanned by the columns of the submatrix $X_S$.
\def\assName{{\bf\rm B}}
\setcounter{Assumption}{-1}
\begin{Assumptions}
  \assumption  There exist $t_0>0$ and a constant $c>0$, such that for all $S$, $|S|=s_0-1$, for all $j\in \bar{S} \cap S^c$, $j' \in {\bar{S}^c} \cap S^c$,  and for a large enough $n$
\begin{equation}
\label{Eq:CondB1}
\frac{1}{n}\left[||\mathcal{P}_{S^\star_j} \mu||_2^2-||\mathcal{P}_{S^\star_{j'}}\mu||_2^2\right] > 4t_0\log c,
\end{equation}
where $S^\star_r \equiv S \cup \{r\}$.
\end{Assumptions}
We note that since $c$ could be lower than one the right hand side of \eqref{Eq:CondB1} can be negative. The following theorem gives conditions under which the simulated annealing algorithm is passing through an interesting model $\bar{S}$. More accurately, the theorem states that there is always strictly positive probability to pass through $\bar{S}$ in the next few moves. This result should apply for all models that Assumption (B1) holds for. Note however, that we do not claim that the algorithm finds all the models in a minimal class. Proving such a result would probably require complicated assumptions on models with larger size than $s_0$, and their relation to $\bar{S}$ and other interesting models.

 Let $P_t^m(S'|S)$ be the probability of passing through model $S'$ in the next $m$ iterations of the algorithm,  given the current temperature is $t$, and the current state of the algorithm is the model $S$.
\begin{theorem}\label{Thm:Algo}
Consider the simulated annealing algorithm with $\kappa=s_0$ and with a $\gamma$ vector such that $\gamma_{min}\ge c_\gamma$. Let Assumptions (A1)-(A4) hold and  let Assumption (B1) hold for some temperature $t_0$ and with $c=c_\gamma$. If $\bar{S} \subseteq A_\gamma$ then for all $S \subseteq A_\gamma$ with $s=s_0$, for all $m\ge s_0-|\bar{S} \cap S|$ and for large enough $n$,
\begin{equation}
\label{Eq:thm}
P_{t_0}^m(\bar{S}|S) > \left[\frac{c^2_\gamma}{s_0(h_\gamma-s_0)}\right]^{s_0}.
\end{equation}
\end{theorem}

A proof is given in the appendix.
Theorem \ref{Thm:Algo} states that for any choice of the vector $\gamma$ such that the entries in $\gamma$ are positive for all the predictors in $\bar{S}$, the probability that the algorithm would visit a $\bar{S}$ in the next $m$ moves is always positive, provided the temperature is high enough, and provided it is possible to move from the current model to $\bar{S}$ in $m$ moves. Recall that our intention here is to use the algorithm as a search algorithm for several models.

For the classical model selection setting with $p<n$, a similar method was suggested by \citeshort{brooks2003classical}. Their motivation is as follows. When searching for the most appropriate model, likelihood based criteria are often used. However, maximizing the likelihood to get parameters estimates for each model becomes infeasible as the number of possible models increases. They therefore suggest to simplify the process by maximizing simultaneously over the parameter space and the model space. They suggest a simulated annealing type algorithm to implement this optimization. The algorithm \citeshort{brooks2003classical} suggested is essentially  an automatic model selection procedure.

\subsection{Choosing $\gamma$}
\label{SubSec:gamma}
The simulated annealing algorithm described above is provided with the vector $\gamma$. The values $\gamma_1,...,\gamma_p$ should represent the knowledge regarding the importance of the predictors, although  we do not assume that any prior knowledge is available. As it can be seen in equations \eqref{Eq:ProbOut}-\eqref{Eq:ProbIn}, predictors with high $\gamma$ values have larger probability to enter the model if they are not part of the current model, and lower probability to be suggested for replacement if they are already part of it. Since $p$ is large, we may also benefit if $\gamma$ includes many zeros.

One simple choice of $\gamma$ is to take the absolute values of the univariate correlations of the different predictors with $Y$. We  could also threshold the correlations in order to keep only predictors having large enough correlation (in absolute value) with $Y$. However,  using univariate correlations is clearly problematic since it overlooks the covariance structure of the predictors in $X$.

 Another possibility is to first use the Lasso with a relatively low penalty, and then to set $\gamma_j=|\bL_j|/||\bL||_1$. The idea behind this suggestion is that predictors with large coefficient value may be more important for prediction of $Y$.

However, as discussed in Section \ref{Sec:descrip}, the Lasso might miss some potentially good predictors. It is well known that the Elastic Net may add these predictors to the solution, although it might also add unnecessary predictors. Moreover, it is not clear how to choose $\gamma_j$ using solely the Elastic Net. The Lasso and the Elastic Net estimators are not model selection consistent in many situations. However, for our purpose, combining both methods together may help us get a reservoir of promising predictors.

\citeshort{zou2005regularization} provided  motivation and results that justify the common knowledge that the Elastic Net is better to use with correlated predictors. Since we intend to exploit this property of the Elastic Net, this paper offers an additional theoretical background. We present a more general result later on this section, but for now, the following proposition demonstrates why the Elastic Net tends to include correlated predictors in its model.
\begin{proposition}
\label{Prop:ENsimple}
Define $X$  and $Y$ as before, and define $\bEN$ by \eqref{Eq:ENDef}.
Denote $\rho=(X^{(1)})^T X^{(2)}$.
Assume $|\bEN_1|\ge c_\beta$ for some $c_\beta>0$. If $|\rho|>1-\lambda_2^2c_\beta^2/||Y||^2_2$ then $|\hat{\beta}^{EN}_2|>0$.
\end{proposition}
A proof is given in the appendix. Proposition \ref{Prop:ENsimple} gives motivation for why $\bEN$ has typically a larger model than $\bL$. It also quantifies how much correlated two predictors need to be so the Elastic Net would either include both predictors or none of them.

Going back to our $\gamma$ vector, the next question is how to use the Lasso and the Elastic Net in order to assign  ``scores'' to each predictor.
Let $S_L$ and $S_{EN}$ be the models that correspond to $\bL$ and $\bEN$, respectively.  Define $S_{+}$ for the group of predictors that were part of the Elastic Net model but not part of the Lasso  model and $S_{out}$ for the predictors that were not included in any of them. Note that $S_L\cap S_{+} = S_L\cap S_{out}= S_{+}\cap S_{out}=\emptyset$ and $S_L\cup S_{+} \cup S_{out}$ is $\{1,...,p\}$. Define
\begin{equation*}
\bL_{+}(\delta)=\argmin_\beta\Bigl(\frac{1}{n}||Y-X\beta||_2^2+\lambda\sum\limits_{j=1}^{p}\delta^{\mathbbm{1}\{j\in S_{+}\}}|\beta_j|\Bigr), \qquad \delta \in (0,1),
\end{equation*}
and let $S_{+}^L(\delta)$ be the appropriate model. In this procedure, a reduced penalty is given for predictors that $\bL$ might have missed. Thus, these predictors are encouraged to enter the model, and since they may take the place of others, predictors in $S_L$ that their explanation power is not high enough are pushed out of the model. Note that $\bL_{+}(\delta)$ is a special case of $\bL_w$, as defined in \eqref{Eq:WeiLasDef}, with $w_j=\delta^{\mathbbm{1}\{j\in S_{+}\}}$.

We demonstrate how the reduced penalty procedure works using a toy example. A data set with $n=30$ and $p=50$ is simulated. The true value of $\beta$ is taken to be $\bz=(0.5\;0.5\;1\;1\;1\;0\; 0 \;...\; 0)^T$ and $\sigma^2$ is taken to be one. The predictors are independent normal variables with the exception of 0.8 correlation between $X^{(1)}$ and $X^{(2)}$. Predictor 1 is included in the Lasso model, however predictor 2 is not. Figure \ref{Fig:toy} presents the coefficients' estimates of $X^{(1)},X^{(2)}$ and $X^{(3)}$ when lowering the penalty of $X^{(2)}$. Note how $X^{(2)}$ enters the model for low enough penalty while $X^{(1)}$ leaves the model for low enough penalty (on $X^{(2)}$).

\begin{figure}[htbp]
\centering
\includegraphics[width = 0.6\textwidth]{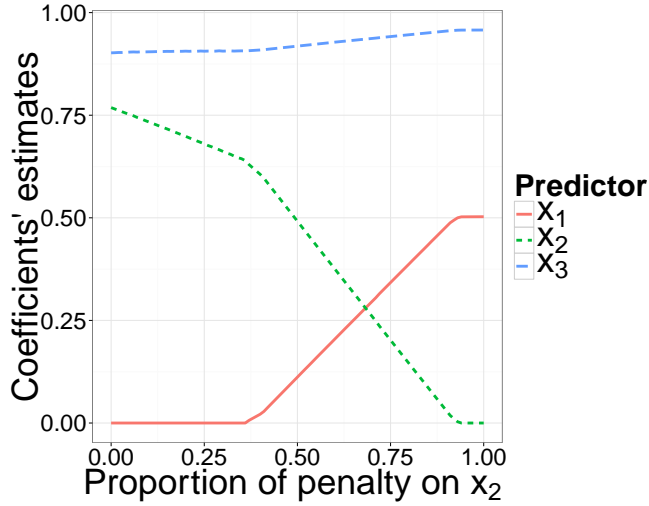}
\caption{Toy example: coefficients' estimates for
\label{Fig:toy}
predictors $X^{(1)},X^{(2)}$ and $X^{(3)}$ when lowering the Lasso penalty for $X^{(2)}$ only. The rightmost point corresponds to a Lasso procedure with equal penalties for all predictors}
\end{figure}

We suggest to measure the importance of a predictor $j\in S_{+}$ by the highest $\delta$ such that $j\in S_{+}^L(\delta)$. On the other hand,  the importance of a predictor $j'\in S_L$, can be measured by the highest $\delta$ such that $j'\notin S_{+}^L(\delta)$ (now, smaller $\delta$ reflects $j'$ is more important). With this in our mind, we continue to the derivation of $\gamma$.

Let $\Delta=(\delta_0<\delta_1<...<\delta_h)$ be some grid of $[0,1]$, with $\delta_0=0$ and $\delta_h=1$. For each $\delta\in \Delta$, we obtain $\bL_{+}(\delta)$. Define
\begin{equation*}
i^\star_j = \left\{
		\begin{array}{ll}
		\underset{i}{\text{argmax}}\{i:\hat{\beta}_{+}^L(\delta_i)_j\ne 0\} & \quad j \notin S_L\\
		\underset{i}{\text{argmax}}\{i:\hat{\beta}_{+}^L(\delta_i)_j=0\}& \quad j \in S_L\\
		\end{array}
\right.
\end{equation*}
and if the $\argmax$ is over an empty set, define $i^\star_j=0$. Let $\delta^j:=\delta_{i^\star_j}$. Now, we suggest to choose $\gamma_j$ as follows:
\begin{equation*}
\gamma_j = \left\{
        \begin{array}{cl}
            0 & \quad j\in S_{out} \\
            \frac{\delta^j}{2} & \quad j\in S_{+}\\
            \\

            1-\frac{\delta^j}{2} & \quad j \in S_{L},
        \end{array}
    \right.
\end{equation*}
for all $j \in \{1,...,p\}$. This choice of $\gamma$ has the following nice properties.
\begin{itemize}
\item A predictor $j\notin S_L$ with $i^\star_j=0$ is excluded from consideration.
\item On the other hand, for a predictor $j\in S_L$, if $i^\star_j=0$ than $\gamma_j=1$, which is the maximal possible value. Even when the penalty for other predictors was dramatically reduced, leading to their entrance to the model, $j$ remains part of the solution and hence it is essential  for prediction of $Y$.
\item  Since predictors in $S_L$ were picked when equal penalty was assigned to all predictors, they get priority over the predictors in $S_{+}$.
\item However, for two identical predictors (or highly correlated predictors)  $X^{(j)}=X^{(j')}$ such that $j\in S_L$ and $j'\notin S_L$, we get a desirable result. By Proposition \ref{Prop:ENsimple} we know that $X_j'\in S_{+}$. Now, for $\delta_{h-1}<1$  it is clear that $j' \in S_{+}^L(\delta_{h-1})$ and $j \notin S_{+}^L(\delta_{h-1})$. Therefore $i^\star_{j}=i^\star_{j'}=h-1$, and hence if $\delta_{h-1}$ is taken to be close to one, then $\gamma_{j}\simeq \gamma_{j'}\simeq 0.5$ as one might want.
\end{itemize}
Proposition \ref{Prop:ENsimple} deals with the case of two correlated predictors. In practice, the covariance structure may be much more complicated. Therefore the question arises: can we say something more general on the Elastic Net in the presence of competing models? Apparently we can. Let $M_1$ and $M_2$ be two models, that is, two sets of predictors, that possibly intersect. Assume that the Elastic Net solution chose all the predictors in $M_1$. What can we say about the predictors in $M_2$? Are there conditions on $X_{M_2}$, $X_{M_1}$ and $Y$ such that all the predictors in $M_2$ are also chosen? If the answer is yes (and it is, as Theorem \ref{Thm:ENgroups} states), it justifies our use of the Elastic Net to reveal more relevant predictors. In our case, the relevant predictors are the building blocks of models in $\mathcal{G}$.

In order to reveal this property of the Elastic Net, we analyze $\bEN$, the solution of \eqref{Eq:ENDef}, when assuming all the predictors in $M_1$ have non-zero values. We denote $M^{(-)}$ for $(M_1\cup M_2)^c$, the set of predictors that are not included in $M_1$ or $M_2$ and $\tilde{X}=X_{M^{(-)}}$ for the appropriate submatrix of $X$.  We let $\bEN_{M_1}$, $\bEN_{M_2}$ and $\tilde{\hat{\beta}}^{EN}$ be the coordinates of $\bEN$ that correspond to  $M_1$, $M_2$ and $(M_1\cup M_2)^c$, respectively. Then, we show that we can concentrate on $\tilde{Y}=Y-\tilde{X}\tilde{\hat{\beta}}^{EN}$, which is the unexplained residual of $Y$, after taking into account $\tilde{X}$. Finally, we show that both $M_1$ and $M_2$ are chosen by the Elastic Net if the prediction of $\tilde{Y}$ using $M_1$, namely $X_{M_1}\bEN_{M_1}$, projected onto the subspace spanned by the columns of $M_2$ is correlated enough with  $\tilde{Y}$. Formally,
\begin{theorem}
\label{Thm:ENgroups}
Define $\bEN$ as before. Let $M_1$ and $M_2$ be two models with the appropriate submatrices $X_{M_1}$ and $X_{M_2}$. Define  $\tilde{X}$ and $\tilde{Y}$ as before. Define $\bEN_{M_1}$ and $\bEN_{M_2}$ as before. Denote $\mathcal{P}_{M_2}$ for the projection matrix onto the subspace spanned by the columns of $X_{M_2}$. WLOG, assume $|M_2|\le|M_1|$ and that all the coordinates of $\bEN_{M_1}$ are different than zero. Finally, if
\begin{equation}
\label{Eq:CondENgroups}
\tilde{Y}^T\mathcal{P}_{M_2}X_{M_1}\bEN_{M_1} > c_1(\lambda_1,\lambda_2,X_{M_1},\tilde{Y},\bEN_{M_1}),
\end{equation}
then all the coordinates of $|\bEN_{M_2}|$ are different than zero.
\end{theorem}
A proof and a discussion on the technical aspects of condition \eqref{Eq:CondENgroups}  and the constant $c_1$  are given in the appendix.
Theorem \ref{Thm:ENgroups} states that under a suitable condition, predictors belong to at least one of two competing models are chosen by the Elastic Net. In our context, when we have a model $M_1$ with a good prediction accuracy, i.e., $X_{M_1}\bEN_{M_1}$ is close to $\tilde{Y}$, then predictors in any another model $M_2$ which has similar prediction, that is $\mathcal{P}_{M_2}X_{M_1}\bEN_{M_1}$ is also close to $\tilde{Y}$,  would be chosen by the Elastic Net. Hence, these predictors are expected to have a positive value in $\gamma$, and our simulated annealing algorithm would pass through these models, provided the conditions in Theorem \ref{Thm:Algo} are met. Therefore, these models are expected to appear in $\mathcal{G}$.
\section{Numerical Results}
\label{Sec:numerRes}
\subsection{Simulation Study}
\label{SubSec:sims}
We consider a setup in which there are few models one would want to reveal. The following model is used $Y=X\beta+\epsilon$, $\epsilon \sim N(0,I)$ with $\beta_j$ equals to $C$ for $j=1,2,...,6$ and zero for $j>6$. $C$ is a constant chosen to get a desired signal to noise ratio (SNR).
The predictors in $X$ are all \iid $N(0,I)$ with the exception of $X^{(7)}$ and  $X^{(8)}$, which defined by
\begin{align*}
X^{(7)}= \frac{2}{3}[X^{(1)}+X^{(2)}]+\xi_1, && \xi_1 \sim N_n\left(0,\frac{1}{3}I\right) \\
X^{(8)}= \frac{2}{3}[X^{(3)}+X^{(4)}]+\xi_2, && \xi_2 \sim N_n\left(0,\frac{1}{3}I\right)
\end{align*}
where $\xi_1$ and $\xi_2$ are independent. In this scenario, there are 4 models we would like to find: (I) \{1,2,3,4,5,6\}; (II)  \{5,6,7,8\};  (III) \{3,4,5,6,7\};  and (IV)  \{1,2,5,6,8\}.

For each simulated dataset, we do the following:
\begin{enumerate}[1.]
\item Obtain $\gamma$ as explained in Section \ref{SubSec:gamma}. The tuning parameter of the Lasso is taken to be the minimizer of the cross-validation MSE. For the Elastic Net, $\alpha$ in \eqref{Eq:ENDefAlp} is taken to be $0.4$.
\item Run the simulated annealing algorithm for $\kappa=4,5,6$. The tuning parameters of the algorithm are chosen quite arbitrarily: $T=10 \times (0.7^1,0.7^2,...,\allowbreak 0.7^{20})$; $\Delta=(0, 0.02, 0.04, ..., 0.98,1)$; $N_t=N=100$ for all $t \in T$.
\item Then, for each model (I)--(IV), we check whether the model is the best model obtained (as measured by MSE) among models with the same size. For example, we check if Model (II) is the best model out of all models that were found with $\kappa=4$. We also check whether the model is one of the top five models among models with the same size.
\end{enumerate}
A 1000 simulated datasets were generated for each different scenario: For $n=100$, $p=200,500,1000$ and for $\text{SNR}=1,2,4,8,12,16$. Table \ref{Tab:SimRes} displays the proportion of times each model was chosen, either as the best one, or as one of the top five models. The results are as one might expect. For large SNR, the models are chosen more frequently. However,  models (III) and (IV) are competing, in the sense that they both include five predictors. Even for large SNR, each of the models, (III) and (IV), are chosen in about $50\%$ of the cases. As recommended in Section \ref{SubSec:annealAlgo}, we should start the algorithm from different initial points, that is, different initial models.
\begin{table}[ht!]
\caption{Proportion that each model is chosen as best model or as one of top five models for different number of potential predictors ($p$) and various SNR values.}
\label{Tab:SimRes}
\centering
\begin{tabular}{|c|c|cc|cc|cc|}
  \hline
   & & \multicolumn{2}{c}{$p=200$}&\multicolumn{2}{|c|}{$p=500$}& \multicolumn{2}{c|}{$p=1000$} \\
   SNR & Model & Best & Top 5 & Best & Top 5 & Best & Top 5 \\
  \hline
    \multirow{4}{*}{1} & (I) & 0.00 & 0.01
  &  0.00 & 0.00 & 0.00 & 0.00 \\
 & (II)& 0.42& 0.62 & 0.28 & 0.46 & 0.23 & 0.38 \\
 & (III)& 0.04& 0.08 & 0.01 & 0.02 & 0.00 & 0.00\\
 & (IV)& 0.04& 0.08 & 0.02 & 0.03 & 0.00 & 0.01\\
   \hline
   \multirow{4}{*}{2} & (I) & 0.10 & 0.12
     &  0.05 & 0.06 & 0.04 & 0.05 \\
    & (II)& 0.94& 0.96 & 0.92 & 0.94 & 0.94 & 0.95 \\
    & (III)& 0.27& 0.34 & 0.18 & 0.24 & 0.15 & 0.17\\
    & (IV)& 0.28& 0.37 & 0.18 & 0.22 & 0.14 & 0.17\\
      \hline
    \multirow{4}{*}{4} & (I) & 0.38 & 0.38
         &  0.20 & 0.20 & 0.11 & 0.11 \\
        & (II)& 0.96& 0.96 & 0.96 & 0.96 & 0.96 & 0.95 \\
        & (III)& 0.38& 0.46 & 0.31 & 0.36 & 0.22 & 0.24\\
        & (IV)& 0.39& 0.46 & 0.28 & 0.31 & 0.24 & 0.26\\
          \hline
      \multirow{4}{*}{8} & (I) & 0.72 & 0.72
            &  0.46 & 0.46 & 0.32 & 0.32 \\
                            & (II)& 0.97& 0.97 & 0.97 & 0.97 & 0.96 & 0.96 \\
                            & (III)& 0.41& 0.48 & 0.36 & 0.40 & 0.30 & 0.31\\
                            & (IV) & 0.44 & 0.50 & 0.34 & 0.37 & 0.29 & 0.31\\
                              \hline
     \multirow{4}{*}{12} & (I) & 0.86 & 0.86
        &  0.66 & 0.66 & 0.49 & 0.49 \\
                  & (II)& 0.98& 0.98 & 0.97 & 0.97 & 0.96 & 0.96 \\
                  & (III)& 0.49& 0.55 & 0.41 & 0.44 & 0.32 & 0.34\\
                  & (IV)& 0.42 & 0.48 & 0.37 & 0.40 & 0.32 & 0.34\\
                    \hline
\end{tabular}
\end{table}
Figure \ref{Fig:mainRes} presents comparison between running the algorithm once and three times, from different points.  Note the improved results for models (III) and (IV) when we start the algorithm from three different starting points.
\begin{figure}[htbp]
        \centering
 \includegraphics[width=0.8\textwidth]{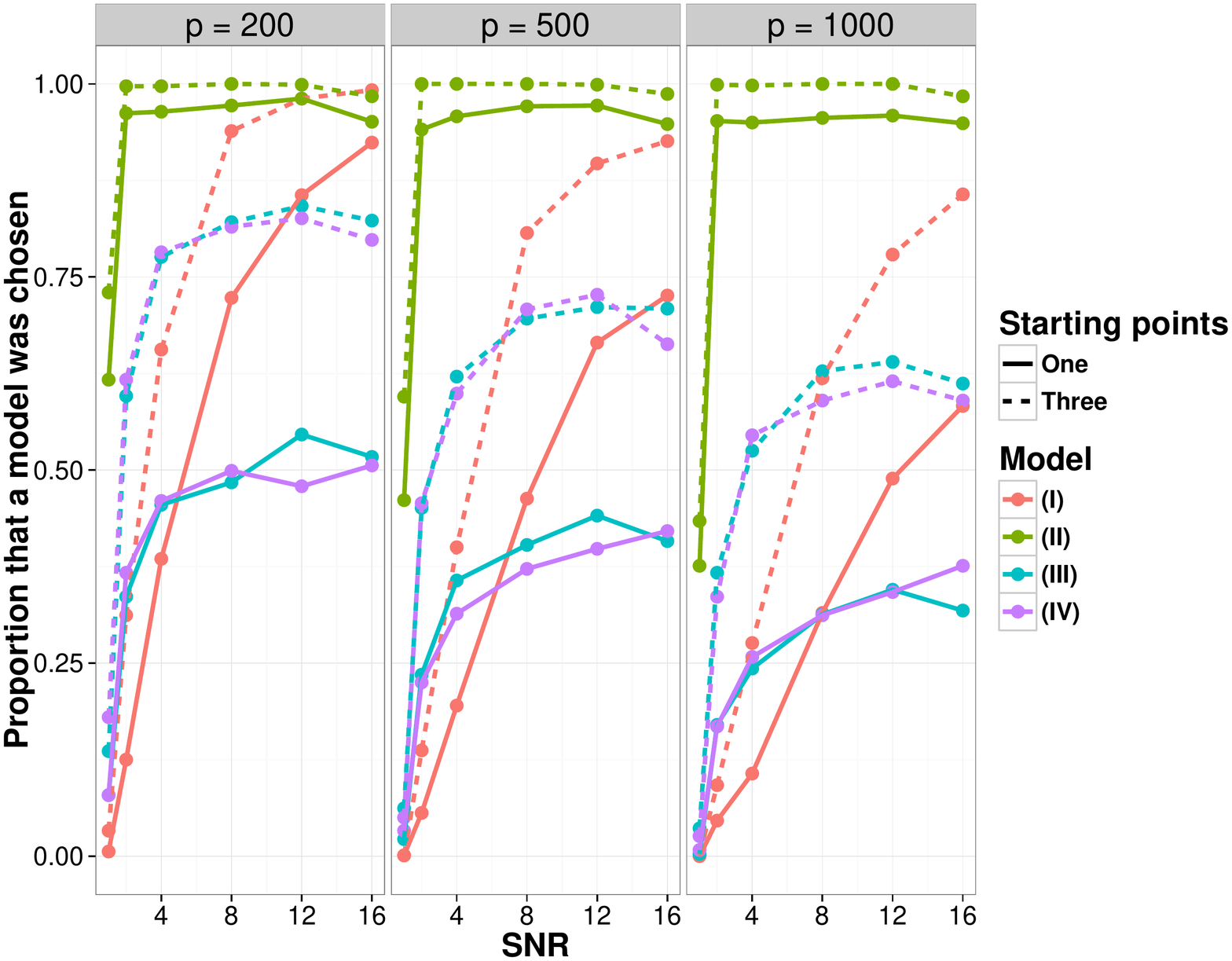}
         \caption{Proportion that each model is chosen as one of top five models for different number of potential predictors ($p$) and various SNR values. There is an apparent improvement when running the algorithm from three starting points.}
         \label{Fig:mainRes}
\end{figure}
The results described in this section are quite similar to results  obtained when forming $\mathcal{G(\kappa,\eta)}$ as defined in \eqref{Eq:Gdef}, for each $\kappa=4,5,6$ separately and using an arbitrary small value of $\eta$.
\subsection{Real data sets}
\label{SubSec:data}
We demonstrate the utility of using a minimal class of models in the analysis of two real datasets. The tuning parameters of the Lasso and the Elastic Net were taken to be the same as in Section \ref{SubSec:sims}. The tuning parameters of the simulated annealing algorithm were $T=10 \times (0.7^1,0.7^2,...,0.7^{20})$, $\Delta=(0, 0.01, 0.02, ...,\allowbreak 0.98,  0.99, 1)$, and $N_t=N=100$ for all $t \in T$.
\subsubsection{Riboflavin}
We use a high-dimensional data about the production of riboflavin (vitamin B2) in Bacillus subtilis that were recently published, \citeshort{buhlmann2014high}. The data consist $p=4088$ predictors. These are measures of log expression levels of genes in $n=71$ observations. The target variable is the (log) riboflavin production rate.

$S_L$ included 40 predictors (and intercept),  and $S_{EN}$ included 59 predictors when taking the tuning parameters as described in Section \ref{SubSec:sims}.
In total, we considered 61 different predictors  (i.e., genes). Panel (a) of Figure \ref{Fig:gammaHist} presents the histogram of the positive values in $\gamma$.
\begin{figure}	
	\centering
	\begin{subfigure}[b]{0.4\textwidth}
		\centering
		\includegraphics[width=\textwidth]{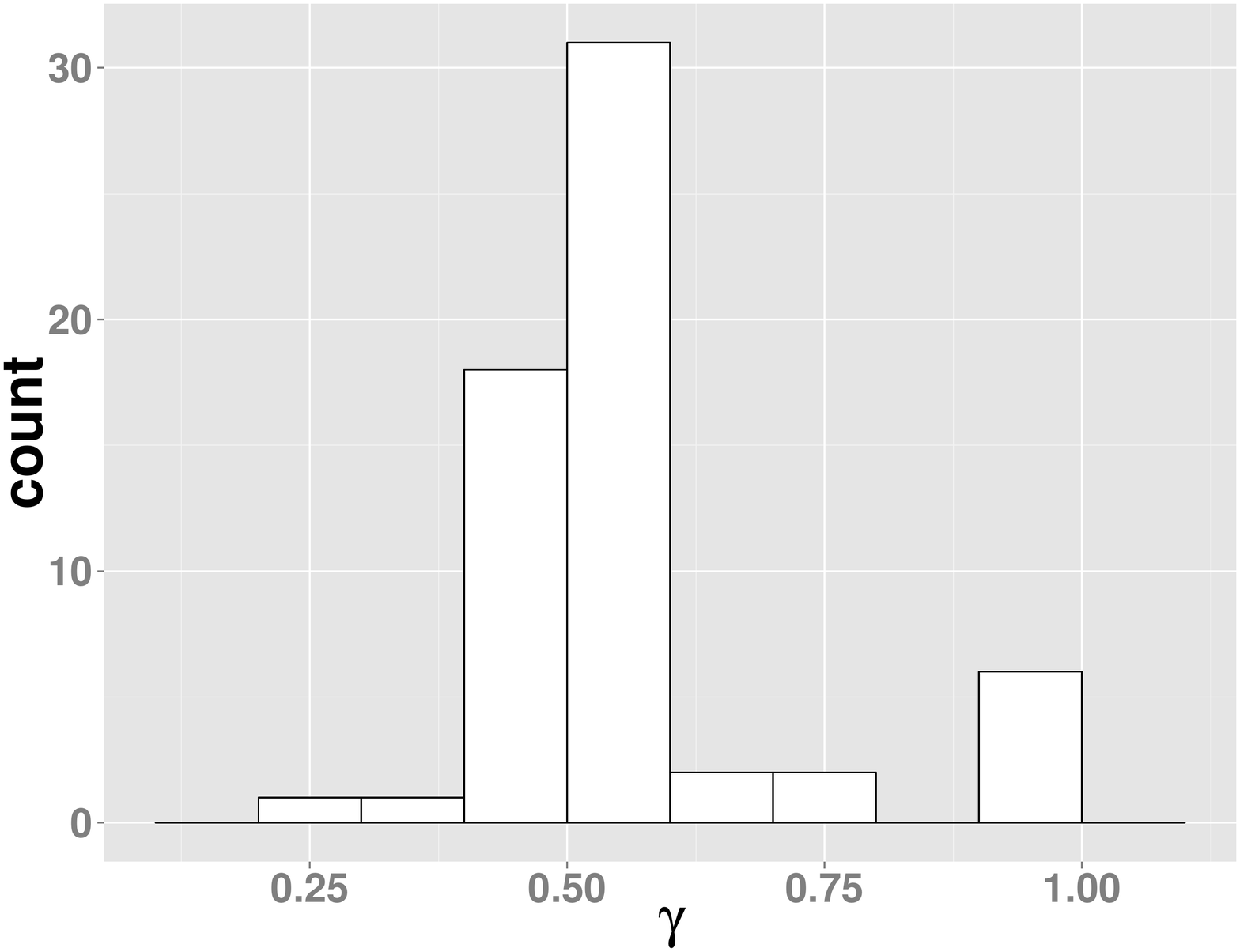}
 \caption{Riboflavin data }
                \label{Fig:gammaHist:Ribo}
	\end{subfigure}
	\quad
	\begin{subfigure}[b]{0.4\textwidth}
		\centering
		\includegraphics[width=\textwidth]{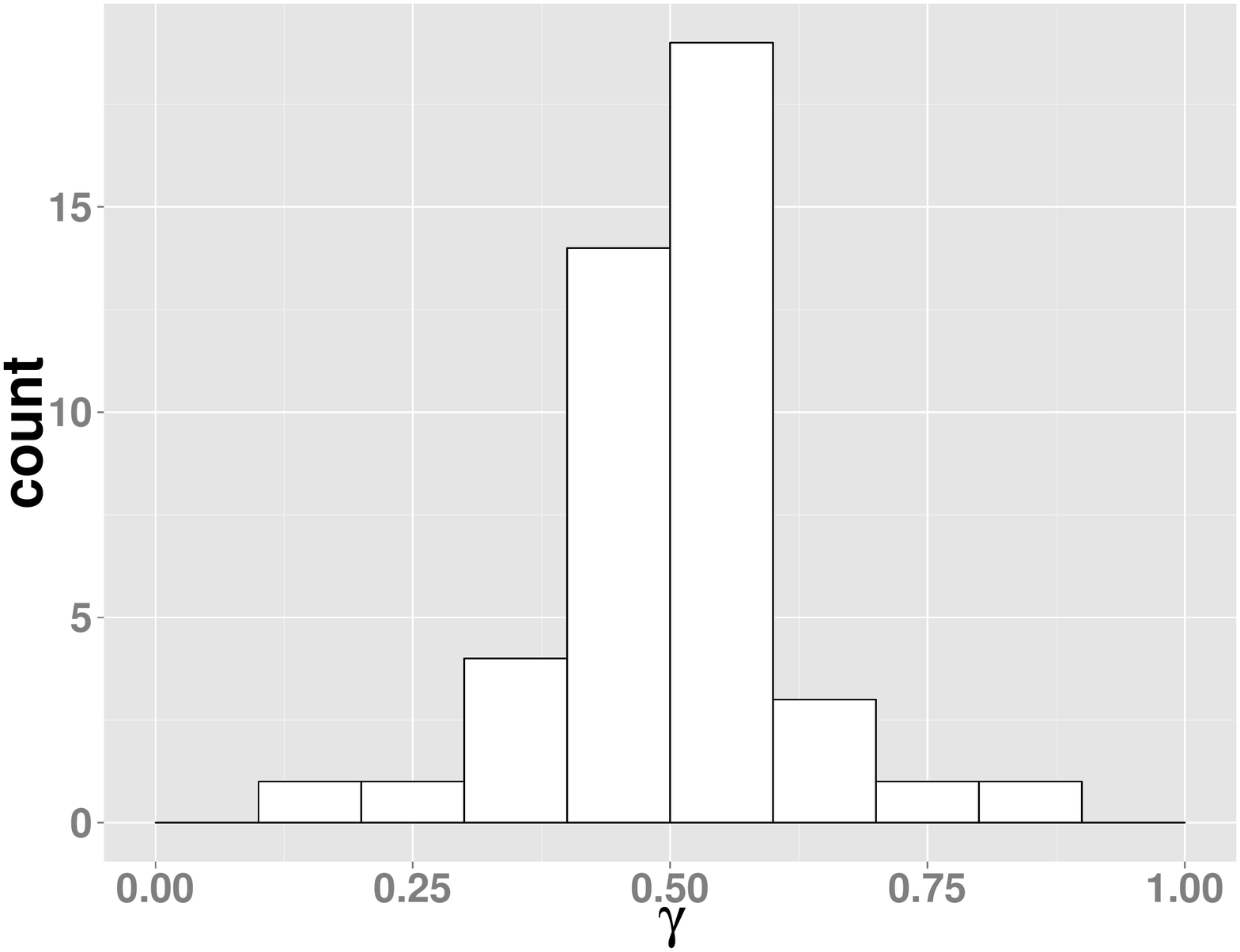}
      \caption{Air pollution data}
				\label{Fig:gammaHist:Death}
	\end{subfigure}
	\caption{Histograms of the values of $\gamma$ for positive entries only in the two dataset analysis examples.}
	\label{Fig:gammaHist}
\end{figure}

We run the algorithm from three random starting points for each model size between 1 and 10. We kept the five best models for each size and starting point. We then combined these models to get, after removal of duplicates, a total of 112 models. See Table \ref{Tab:RiboModelSizes} for the number of unique models as a function of the model size.
The following insights are drawn from examining more carefully the  models we obtained (see Table \ref{Tab:RiboModels} in the appendix):
\begin{itemize}
\item In total, the models include 53 different predictors. Out of these, 35 predictors appear in less than $10\%$ of the models, meaning they are probably less important as predictors of riboflavin production rate.
\item  Gene number $2564$ appears in all models of size larger than 3 and in 5 out of 8 models of size 3. However,  this gene is not included in any of the smaller models. This gene is the only one that appears in more than half of our models. We can infer that while this gene does not hold an effect strong enough comparing to other genes in order to stand out, it has a unique relation with the outcome predictor that could not be mimicked using other combination of genes.
\item At least one gene from the group $\{4002,4003,4004,4006\}$ is contained in all models of size larger than one, although never more than one of these genes. Genes number $4003$ and $4004$ appear more frequently than genes number $4002$ and $4006$. Looking at the correlation matrix of these genes only, we see they are all highly correlated (pairwise correlations $>0.97$). Future research could take this finding into account  by using, e.g., the Group Lasso, \citeshort{yuan2006model}.
\item Similarly,  either gene number $1278$ or gene number $1279$ appear in about half of the models. They are also strongly correlated $(0.984)$. The same statement holds for genes number $69$ and $73$ (correlation of $0.945$) as well.
\item The impotence of genes number $792$,$1131$, and possibly others, should be also examined since each of them appears in a variety of different models.
\end{itemize}
\begin{table}
\centering
\caption{Riboflavin data: Number of unique models for each model size after running the algorithm from 3 different starting points}
\label{Tab:RiboModelSizes}
\begin{tabular}{|c|cccccccccc|}
\hline \rule[-1ex]{0pt}{4.5ex}
Model size & 1 & 2 & 3 & 4 & 5 & 6 & 7 & 8 & 9 & 10 \\
\hline \rule[-1ex]{0pt}{4.5ex}
Number of models & 5 & 5 & 8 & 6 & 13 & 15 & 15 & 15 & 15 & 15 \\
\hline
\end{tabular}
\end{table}
We now compare our results to models obtained using other methods, as reported in \citeshort{buhlmann2014high}. The multiple sample splitting method to get $p$-values, \citeshort{meinshausen2009p},  yields only one significant predictor. Indeed, a model that includes only this predictor is part of our models. If one constructs his model using the stability selection, \citeshort{meinshausen2010stability}, as a screening process for the predictors, he would get a model consisting three genes, which correspond to columns number $625,2565$ and $4004$ in our $X$ matrix. However, this model is not included in our top models. In fact, the highest MSE for a model in our 8 models of size 3 is 0.2047 while the MSE of the model suggested using the stability selection is 0.2703, more than $30\% $ difference!
\subsubsection{Air pollution}
\label{SubSubSec:mortData}
We now demonstrate how the proposed procedure can be used for traditional, purportedly simpler, problem. The air pollution data set, \citeshort{mcdonald1973instabilities}, includes 58 Standard Metropolitan Statistical Areas (SMSAs) of the US (after removal of outliers). The outcome variable is age-adjusted mortality rate. There are 15 potential predictors including air pollution, environmental, demographic and socioeconomic predictors. Description of the predictors is given in Table \ref{Tab:mortPredictors} in the appendix.

There is no guarantee that the relationship between the predictors and the outcome variable has linear form. We therefore include commonly used transformations of each variable, namely natural logarithm, square root and power of two transformations. Considering also all possible two way interactions, we have a total of 165 predictors.

High-dimensional regression model that includes transformations and interactions has been dealt with in the literature. For example, by  using two steps procedures, e.g., \citeshort{bickel2010hierarchical}, or by solving a relevant optimization problem, e.g.,  \citeshort{bien2013lasso}. Our procedure has a different goal, since we are not looking for the best predictive model, but rather for a meaningful insights about the data.

Following the Lasso and Elastic Net step, we are left with 44 predictors with positive $\gamma_j$ (one untransformed predictor, 3 log transformations, 4 square root transformations, 8 power of two transformations and the rest are interactions). Panel (b) of Figure \ref{Fig:gammaHist} presents the histogram of the positive values in $\gamma$.

For each $\kappa =1,2,...,10$, we run the algorithm from three starting points, and then keep the 5 best models. In total, we get 126 unique models. Table \ref{Tab:DeathVars} summarizes the results for prominent predictors, that is, predictors that appear in at least quarter of the models we obtained. The table presents a matrix of the joint frequency of each two predictors. Each cell in the table is the number of models including both the predictor listed in the row and the predictor listed in the column. The diagonal is simply the number of models that a predictor appears in.
\begin{table}[h!]
\caption{Frequency that each two predictors together in the 126 models. The diagonal is simply the number of models that a predictor appears in. For example, in 27 models both $\log(\textbf{NOx})$ and $\sqrt{\textbf{nwht}}$ appear}
\label{Tab:DeathVars}
\centering
\begin{tabular}{|l|cccccccc|}
  \hline
 &  (1)  & ( 2) &(3) & (4) & (5) &(6)  &(7) &(8) \\
  \hline
(1) $\log(\textbf{NOx})$ &  97 &  27 &  36 &  30 &  50 &  31 &  33 &  35 \\
(2)  $\sqrt{\textbf{nwht}}$ &  &  33 &   7 &   8 &  14 &  10 &   0 &   0 \\
(3)  $\sqrt{\textbf{HC}}$ &  &  &  37 &  12 &  18 &  10 &  16 &  15 \\
(4)  $\textbf{HC}\times \textbf{prec}$ &  &  &  &  33 &  26 &  17 &  18 &   8 \\
(5)  $\textbf{jant}\times \textbf{ovr65}$ &  &  &  &  &  66 &  30 &  27 &  26 \\
(6)  $\textbf{pphs}\times \textbf{educ}$  &  &  &  &  &  &  37 &  14 &  14 \\
(7)  $\textbf{nwht}\times \textbf{ofwk}$ &  &  &  &  &  &  &  46 &   1 \\
(8) $\textbf{nwht}\times \textbf{mst}$ &  &  &  &  &  &  &  &  43 \\
   \hline
\end{tabular}
\end{table}

Three (transformed) main effects are chosen. The nitric oxide pollution is invaluable for prediction of mortality rate. This predictor (in a log shape) appears in a large majority of the models. Apart from this predictor, the hydrocarbon pollution appears (after a square root transformation), but only in about $30\%$ of the models. There is, however,  one result that catches the eye. The two zeros in the matrix (second row, last two values) mean that interactions involving the percentage of non-white population are only part of models that do not include the percentage of non-white population as a main effect. Moreover, the two interactions do not make much sense. The evident conclusion is that the two interactions took the place of the main effect. We therefore repeat the analysis after the removal of these two interactions.

The new frequency matrix is displayed in Table  \ref{Tab:DeathVarsSt2}. The conclusion regarding the importance of the nitric oxide pollution remains. Nevertheless, hydrocarbon pollution is not relevant anymore. The percentage of non-white population appears untransformed but also after taking its squared root. However, this predictor appears in single form only for each model. We conclude that this predictor should be used for prediction of the mortality rate, but the question of transformation remains unsolved.
\begin{table}[h!]
\caption{Frequency that each two predictors appear together in the 126 models obtained after removal of the two interactions. The diagonal is simply the number of models that a predictor is included. For example, in 43 different models  both $\log(\textbf{NOx})$ and $\sqrt{\textbf{nwht}}$ appear.}
\label{Tab:DeathVarsSt2}
\centering
\begin{tabular}{|l|cccccccc|}
  \hline
& (1)  & (2)  & (3)  & (4)  & (5)  & (6)  & (7)  & (8)  \\
  \hline
(1)  $\textbf{nwht}$ &  37 &  13 &  36 &   0 &  12 &  17 &  28 &  26 \\
(2)  $\log(\textbf{prec})$ &  &  31 &  31 &  13 &   9 &   2 &  15 &  19 \\
(3)  $\log(\textbf{NOx})$ &  &  & 106 &  43 &  28 &  40 &  67 &  62 \\
(4)   $\sqrt{\textbf{nwht}}$ &  &  &  &  44 &  11 &  13 &  24 &  24 \\
(5)  $\textbf{dens}\times \textbf{prec}$ &  &  &  &  &  30 &  14 &  26 &  20 \\
(6)  $\textbf{hum}\times \textbf{prec}$  &  &  &  &  &  &  40 &  38 &  28 \\
(7)  $\textbf{jant}\times \textbf{ovr65}$ &  &  &  &  &  &  &  68 &  47 \\
(8)  $\textbf{pphs}\times \textbf{educ}$ &  &  &  &  &  &  &  &  63 \\
   \hline
\end{tabular}
\end{table}

Turning to the interactions. The interaction between percentage of elderly population  and the average temperature in January appears while the appropriate main effects do not appear. However, the absence of age related effect is not so surprising since the outcome variable, the mortality rate,  is age corrected. The interaction between the household size and the level of education appears in half of the models, whereas appropriate main effects do  not appear. This interaction could be a proxy to other effects that were not measured. Interactions involving the average precipitation appear less than other predictors. The interaction with humidity usually appears without the main effect of precipitation. Nevertheless, both interactions should be taken into account when constructing a prediction model for the mortality rate.
\section{Discussion}
\label{Sec:disc}
Model selection consistency is an ambitious goal to achieve when dealing with high-dimensional data. A ``minimal class of models'' was defined to be a set of models that should be considered as candidates for prediction of the outcome variable. A search algorithm to identify these models was developed using simulated annealing method. Under suitable conditions, that are outlined in Theorem \ref{Thm:Algo}, the algorithm passes through models of interest.

A score for each predictor is given using  the Lasso, the Elastic Net and a reduced penalty Lasso. These scores are used by the search algorithm. They are not necessarily optimal  but we claim that they are sensible. Other scoring methods may achieve better results. On the other hand, the scores we use here may be used for other purposes. Theoretical justification for using the Elastic Net to unveil predictors the Lasso might have missed was also presented.

 A simulation study was conducted to demonstrate the capability of the search algorithm to detect relevant models. As illustrated using real data examples, a class of minimal models can be used to derive conclusions regarding the problem at hand. This is rarely the case that a researcher believes a one true model exists, especially in the $p>n$ regime. Therefore, we suggest to abandon the search for this ``holy grail'', and to analyze the class of minimal models instead.

 It is well known that achieving good prediction and successful model selection simultaneously, in a reasonable computation time, is impossible, especially in the high-dimensional setting. We therefore suggested here to make a compromise. Our approach is not necessarily optimal for prediction, nor for model selection. However, it offers a data analysis method that takes into account the uncertainty in model selection, but ensures reasonable prediction accuracy. This method can be used for either prediction, parameter estimation or model selection.

\appendix
\vspace*{2em}
\LARGE
\noindent
\textbf{Appendix}
\normalsize
\section{Proofs}

\renewcommand{\thesubsection}{A.\arabic{subsection}}
\subsection{Proof of Theorem \ref{Thm:Algo}}

We start with the following lemma.
 \begin{lemma}
\label{Lem:submodel}
Assume $Y=\mu+\epsilon$ and assume also (A1)-(A4). Let $\mathcal{S}_k=\{S: |S|=k, \hat{\beta}_S=(X_S^T X_S)^{-1}X_S^T Y\}$ be the set of all models with $k$ variables, such that $\hat{\beta}_S$, the LS estimate, is unique. Denote $S^\star_j=S\cup \{j\}, j \notin S$ for a model that includes $S$ and additional variable $j$ not in $S$. We have
\begin{equation*}
\max_{\substack{ S\in \mathcal{S}_k\\  1\le j\le p}} \epsilon^T(X_{S^\star_j}\hat{\beta}_{S^\star_j}-X_S\hat{\beta}_S)=\o_p(n)
\end{equation*}
\end{lemma}
\begin{proof}
Let $\xi_j$ be the vector of coefficients obtained by regressing $X^{(j)}$, the $j^{th}$ column in $X$, on $X_S$ and let $\mathcal{P}_j$ be the projection operator on the subspace spanned by the part of $X^{(j)}$ which is orthogonal to the subspace spanned by $X_S$. That is,
\begin{equation*}
\mathcal{P}_j=\frac{(X^{(j)}-X_S\xi_j)(X^{(j)}-X_S\xi_j)^T}{||X^{(j)}-X_S\xi_j||_2^2}.
\end{equation*}
Let $\hat{\beta}_{S^\star_j}^j$ be the coefficient estimate of $X^{(j)}$ in model $S^\star_j$, and let $\hat{\beta}_{S^\star_j}^{-j}$ be the coefficient estimates of the variables in $S$ but for the model $S^\star_j$. Since $(X^{(j)}-X_S\xi_j)$ is orthogonal to the subspace spanned by the columns of $X_S$ we have
\begin{align*}
X_{S^\star_j}\hat{\beta}&_{S^\star_j}=X^{(j)}\hat{\beta}_{S^\star_j}^j+X_S\hat{\beta}_{S^\star_j}^{-j}\\
&=(X^{(j)}-X_S\xi_j)\hat{\beta}_{S^\star_j}^j+X_S(\hat{\beta}_{S^\star_j}^{-j}+\xi_j\hat{\beta}_{S^\star_j}^j) \\
&= (X^{(j)}-X_S\xi_j)\hat{\beta}_{S^\star_j}^j+X_S\hat{\beta}_S\\
&= \mathcal{P}_jy+X_S\hat{\beta}_S.
\end{align*}
Therefore,
\begin{equation*}
\epsilon^T (X_{S^\star_j}\hat{\beta}_{S^\star_j}-X_S\hat{\beta}_S) = \epsilon^T\mathcal{P}_j\mu+\epsilon^T\mathcal{P}_j\epsilon.
\end{equation*}
Now, since $||\mathcal{P}_j\mu||_2^2\le||\mu||_2^2=\O(n)$, we get that for all $j$,  $\epsilon^T\mathcal{P}_j\mu=O_p(\sqrt{n})$. Next, let $Z_1,...,Z_{p^{k+1}}$ be $N(0,\sigma^2)$ random variables and observe that the approximate size of the set $\{S_k\} \times \{1,...,p\}$ is $p^{k+1}$. We have for any $a>0$
\begin{equation*}
P\left(\max_{\substack{ S\in \mathcal{S}_k\\  1\le j\le p}}\frac{1}{n}\epsilon^T\mathcal{P}_j\epsilon \ge a \right)\le P\left(\max_{ 1\le j\le p^{k+1}}|Z_j| \ge \sqrt{\frac{an}{\sigma^2}}\right) \le \sigma \sqrt\frac{2(k+1)\log p+\o(1)}{an}.
\end{equation*}
Now, since $p=n^\alpha$ and $k=o(n/\log n)$ we get that
\begin{equation*}
P\left(\max_{\substack{ S\in \mathcal{S}_k\\  1\le j\le p}}\frac{1}{n}\epsilon^T\mathcal{P}_j\epsilon \ge a \right)=o(1)
\end{equation*}
and we are done.
\end{proof}
We can now move to the proof of Theorem \ref{Thm:Algo}. For simplicity, the notation of $i$ as the iteration number for the current temperature $t$ is suppressed.
Note that it is enough to only consider models such that $S \cap \bar{S} = \emptyset$ and to consider $m=s_0$.
Denote $Q_t(S,g,j)$ for the probability of a move in the direction of $\bar{S}$ in the next iteration, that is, the probability of choosing a variable $j \in S\cap \bar{S}^c$ and replace it with a variable $g \in S^c \cap \bar{S}$.  Denote $S'=\{S/\{j\}\}\cup\{g\}$ for this new model. We have
\begin{equation}
\label{Eq:probQ}
\begin{aligned}
Q_t&(S,g,j)\\
&=P(S\rightarrow S')\min\left[1,\exp\left(\frac{||Y-X_S\hat{\beta}_{S}||^2_2-||Y-X_{S'}\hat{\beta}_{S'}||^2_2}{t}\right)\frac{P(S'\rightarrow S)}{P(S\rightarrow S')}\right]
\end{aligned}
\end{equation}
where $P(S \rightarrow S')$ is the probability of suggesting $S'$, given current model is $S$.
Now, since $\gamma_{min} \ge c_\gamma $ and since the maximal value in $\gamma$ equals to one by definition, we have for all $S \subseteq A_\gamma$,
\begin{align}
\label{Eq:proofbounds}
c_\gamma(h_\gamma-s_0)&\le \sum_{u \notin S}\gamma_u \le h_\gamma-s_0 \nonumber\\
s_0& \le  \sum_{v \in S}\frac{1}{\gamma_v} \le \frac{s_0}{c_\gamma}.
\end{align}
Now, by substituting \eqref{Eq:proofbounds} into \eqref{Eq:ProbOut}-\eqref{Eq:prob.model.inv} we get
\begin{equation}
 \begin{aligned}
{P(S\rightarrow S')} &= \frac{\gamma_g}{\sum_{u \notin S}\gamma_u} \frac{1/\gamma_j}{\sum_{v\in S} \frac{1}{\gamma_v}} \ge \frac{c_\gamma^2}{s_0(h_\gamma-s_0)},\\
\frac{P(S'\rightarrow S)}{P(S\rightarrow S')} & = \frac{\gamma_j^2}{\gamma_g^2}\frac{\sum_{u \notin S}\gamma_u\sum_{v\in S} \frac{1}{\gamma_v}}{\sum_{u \notin S'}\gamma_u\sum_{v\in S'} \frac{1}{\gamma_v}}\ge c_\gamma^4. \label{Eq:bounds}
\end{aligned}
\end{equation}
Next, we have
\begin{align}
\nonumber
\frac{1}{n}||& Y-X_S\hat{\beta}_{S}||^2_2-\frac{1}{n}||Y-X_{S'}\hat{\beta}_{S'}||^2_2  \\ \nonumber
&= \frac{1}{n}\left[(Y-X_{S'}\hat{\beta}_{S'})+(Y-X_S\hat{\beta}_S)\right]^T\left(X_{S'}\hat{\beta}_{S'}-X_S\hat{\beta}_S\right)\\ \nonumber
&= \frac{1}{n}Y^T\left(X_{S'}\hat{\beta}_{S'}-X_S\hat{\beta}_S\right)\\ \nonumber
&= \frac{1}{n}\mu^T\left(X_{S'}\hat{\beta}_{S'}-X_S\hat{\beta}_S\right)+\frac{1}{n}\epsilon^T\left(X_{S'}\hat{\beta}_{S'}-X_S\hat{\beta}_S\right)\\ \label{Eq:mu.delta}
&=\frac{1}{n}\mu^T\left(X_{S'}\hat{\beta}_{S'}-X_S\hat{\beta}_S\right)+\Delta_n(S,S')
\end{align}
where the second equality is due to $\hat{\beta}_S$ and $\hat{\beta}_{S'}$ being LS estimators.  We get that an estimator in linear model achieves better (lower) sample MSE, if the correlation of the prediction using this estimator with $Y$ is larger. Now, denote $S''=S'\cup S$. We have
\begin{equation*}
\Delta_n(S,S') = \frac{1}{n}\epsilon^T\left[(X_{S''}\hat{\beta}_{S''}-X_S\hat{\beta}_S)-(X_{S''}\hat{\beta}_{S''}-X_{S'}\hat{\beta}_{S'})\right]
\end{equation*}
and if we apply Lemma \ref{Lem:submodel} twice we get that $\Delta_n(S,S')=\o_p(1)$. Now, regarding the first term in \eqref{Eq:mu.delta},
\begin{equation}
\begin{aligned}
\label{Eq:mu.delta.first.term}
\frac{1}{n}\mu^T&\left(X_{S'}\hat{\beta}_{S'}-X_S\hat{\beta}_S\right) = \frac{1}{n}\mu^T\left(\mathcal{P_{S'}}y-\mathcal{P_S}y \right)  \\
&=\frac{1}{n}\left(||\mathcal{P_{S'}}\mu||_2^2-||\mathcal{P_{S}}\mu||_2^2\right)+\Delta'_n(S,S')
\end{aligned}
\end{equation}
where $\Delta'_n(S,S')=\frac{1}{n}\mu^T\left[\mathcal{P_{S'}}\epsilon-\mathcal{P_{S}}\epsilon\right]$.  The content of the proof of Lemma \ref{Lem:submodel} implies that $\Delta'_n(S,S')=\o_p(1)$. Now, by \eqref{Eq:mu.delta} and \eqref{Eq:mu.delta.first.term} and since Assumption (B1) holds for $t_0$  we get that for large enough $n$
\begin{equation}
\begin{aligned}
\label{Eq:minlargeone}
\frac{1}{n}\left(||Y-X_S\hat{\beta}_{S}||^2_2-||Y-X_{S'}\hat{\beta}_{S'}||^2_2\right)
\ge 4t \log c_\gamma.
\end{aligned}
\end{equation}
Now, by substituting \eqref{Eq:bounds} and \eqref{Eq:minlargeone} into \eqref{Eq:probQ} we get that for large enough $n$,
\begin{equation*}
Q_{t_0}(S,g,j) \ge  \frac{c_\gamma^2}{s_0(h_\gamma-s_0)}
\end{equation*}
for all $S \ne \bar{S}$, $j \in S\cap \bar{S}^c$ and $g \in S^c \cap \bar{S}$. \eqref{Eq:thm} follows from this immediately since for any integer $m$ and for all $S \ne \bar{S}$,
\begin{equation*}
P_{t_0}^m(S'|S) \ge \min_{\substack{\{S:S\cap \bar{S}=\emptyset\} \\j \in S\cap \bar{S}^c\\g \in S^c \cap \bar{S}}}[Q_t(S,g,j)]^{s_0} \ge \left[\frac{c^2_\gamma}{s_0(h_\gamma-s_0)}\right]^{s_0}.
\end{equation*}
\subsection{Proof of Proposition \ref{Prop:ENsimple}}
Recall that the Elastic Net estimator $\bEN$ minimizes
\begin{equation}
\label{Eq:elasticfunc}
||Y-X\beta||^2_2+\lambda_1|\beta|+\lambda_2||\beta||^2_2
\end{equation}
Now, WLOG assume that $\hat{\beta}^{EN}$ is a solution such that $\bEN_1>0$. For convenience, we omit the ``$EN$'' superscript from now on (i.e., $\hat{\beta}=\bEN$). Define the subspace
\begin{equation}
\label{Eq:subspaceB}
\mathcal{B}:=\{\beta: \forall{i \ne 1,2} \hspace{2mm} \beta_i=\hat{\beta}_i, \hspace{2mm} \beta_1=\tau\hat{\beta}_1, \hspace{2 mm} \beta_2=(1-\tau)\hat{\beta}_1 \}.
\end{equation}
If the minimum of \eqref{Eq:elasticfunc} over $\mathcal{B}$ is obtained for $\tau\ne1$, then given that predictor $1$ is part of the Elastic Net model, predictor $2$ is also part of this model.

WLOG, write down $X$ as $X=(X_{(12)}\quad X_{-(12)})$ where $X_{(12)}=(X^{(1)} \quad X^{(2)})$ are the first two columns of $X$ and $X_{-(12)}$ are the rest of its columns. Similarly, we have $\beta^T = (\beta_{(12)}^T \quad \beta_{-(12)}^T)$ where $\beta_{(12)}$ is the first two entries in the vector $\beta$ and $\beta_{-(12)}$ is the rest of the vector. Define $\tilde{Y}=Y-X_{-(12)}\beta_{-(12)}$. We can rewrite \eqref{Eq:elasticfunc} as
\begin{equation}
\label{Eq:elasticfunctild}
||\tilde{Y}-X_{(12)}\beta_{(12)}||^2_2+\lambda_1(|\beta_{-(12)}|+|\beta_{(12)}|)+\lambda_2(||\beta_{-(12)}||^2_2+||\beta_{(12)}||^2_2)
\end{equation}
If the minimum of \eqref{Eq:elasticfunctild}, on  $\mathcal{B}$,  is achieved at $0<\tau^*<1$ then $\hat{\beta}_2$ must be non zero. Minimizing \eqref{Eq:elasticfunctild} on $\mathcal{B}$ is essentially minimizing
\begin{equation}
\label{Eq:whattomin}
-2\tilde{Y}^T X_{(12)}\beta_{(12)}+||X_{(12)}\beta_{(12)}||^2_2+\lambda_2||\beta_{(12)}||^2_2
\end{equation}
on $\mathcal{B}$ .
Now, by the definition of $\mathcal{B}$ in \eqref{Eq:subspaceB} and using simple algebra we get that \eqref{Eq:whattomin} equals to
\begin{equation*}
2\left[\hat{\beta}_1\tilde{Y}^T\left(\tau(X^{(2)}-X^{(1)})-X^{(2)}\right)-\hat{\beta}_1^2\tau(1-\tau)(1-\rho)+\lambda_2\hat{\beta}_1^2\left(\frac{1}{2}-\tau(1-\tau)\right)\right].
\end{equation*}
This is a quadratic function of $\tau$, and by equating its derivative to zero we get that
\begin{equation*}
\tau^*=\frac{1}{2}-\frac{\tilde{Y}^T(X^{(2)}-X^{(1)})}{2\hat{\beta}_1(\lambda_2+1-\rho)}
\end{equation*}
is the minimizer of \eqref{Eq:elasticfunc} (the coefficient of the quadratic term is positive). Note that for $X^{(2)}=X^{(1)}$ we get the expected $\tau^*=\frac{1}{2}$ solution. Note also that this reveals no information regarding the Lasso where $\lambda_2=0$.
Next, we get that $0<\tau^*<1$ if
\begin{equation}
\label{Eq:taufrac}
\left|\frac{\tilde{Y}^T(X^{(2)}-X^{(1)})	}{\hat{\beta}_1(\lambda_2+1-\rho)}\right| < 1.
\end{equation}
Since $||X^{(2)}-X^{(1)}||^2_2=2(1-\rho)$ we have
\begin{equation*}
|\tilde{Y}^T(X^{(2)}-X^{(1)})|\le \sum\limits_{i=1}^{n}|\tilde{Y}_i||X^{(2)}_i-X^{(1)}_i|\le {||\tilde{Y}||_2}\sqrt{2(1-\rho)},
\end{equation*}
using the triangle inequality and then Cauchy-Schwartz inequality.
It is assumed  that $\hat{\beta}_1\ge c_\beta>0$ and it is known that  $||\tilde{Y}||_2\le ||Y||_2$. Therefore, we may rewrite \eqref{Eq:taufrac} as
\begin{equation*}
\frac{\sqrt{2}||Y||_2\sqrt{1-\rho}}{c_\beta(\lambda_2+1-\rho)} < 1.
\end{equation*}
Now, Denote $t=\sqrt{1-\rho}, u=\frac{||Y||_2}{c_\beta}$, we have
\begin{equation*}
t^2-\sqrt{2}ut+\lambda_2>0.
\end{equation*}
For $\lambda_2>\frac{1}{2}u^2$, we get the result we want for all $\rho$'s.
For $\lambda_2<\frac{u^2}{2}$ we have
\begin{align}
\sqrt{1-\rho}>\frac{1}{\sqrt{2}}(u+\sqrt{u^2-2\lambda_2})\label{Eq:bad},\\
\sqrt{1-\rho}<\frac{1}{\sqrt{2}}(u-\sqrt{u^2-2\lambda_2})\label{Eq:good}.
\end{align}
The RHS of \eqref{Eq:bad} is larger than $1$ if $\lambda_2<\sqrt{2}u-1$. That is, there is no suitable  $\rho$ for this case. The RHS of \eqref{Eq:good} is always positive, and for the same condition $\lambda_2<\sqrt{2}u-1$, it also meaningful, i.e., $(u-\sqrt{u^2-2\lambda_2})<\sqrt{2}$ and in terms of $\rho$,
\begin{equation*}
\rho>1-\frac{1}{2}(u-\sqrt{u^2-2\lambda_2})^2
\end{equation*}
or alternatively,
\begin{equation*}
\rho>1-\frac{u^2}{2}\left(1-\sqrt{1-\frac{2\lambda_2}{u}}\right)^2
\end{equation*}
and by Taylor expansion for $2\lambda_2/u$ we get
\begin{equation*}
\rho>1-\frac{\lambda^2_2}{2u^2}
\end{equation*}
$\hfill \square$
\subsection{Proof of Theorem \ref{Thm:ENgroups}}
The proof is similar to the proof of Proposition \ref{Prop:ENsimple}. Let $\hat{\beta}=\bEN$ be the Elastic Net estimator and denote $\hat{\beta}_M$ for the values in $\hat{\beta}$ corresponding to the set of predictors $M$. We can partition the set of potential predictors $\{1,2,...,p\}$ to four disjoint subsets: $M^{(-)}$; $M_1\cap M_2^c$; $M_1^c\cap M_2$ and $M_1\cap M_2$. We replace \eqref{Eq:subspaceB} with
\begin{align}
\label{Eq:subspaceBgroups}
\mathcal{B}:=&\{\beta: \beta_{M^{(-)}}=\hat{\beta}_{M^{(-)}} \hspace{2mm} \beta_{M_1\cap M_2}=\hat{\beta}_{M_1\cap M_2}, \\
\qquad &\beta_{M_1\cap M_2^c}=\tau\hat{\beta}_{M_1\cap M_2^c}, \hspace{2mm}  \beta_{M_1^c\cap M_2}=(1-\tau)\Theta'\hat{\beta}_{M_1\cap M_2^c}\}.
\end{align}
where $\beta_M$ is defined as the values in $\hat{\beta}$ corresponding to the set $M$ and
$\Theta'$ is the matrix of coefficients obtained from regressing $X_{M_1\cap M_2^c}$ on $X_{M_1^c\cap M_2}$. We define $\Theta$ to be an augmented version of $\Theta'$, which we obtain by  regressing $X_{M_1}$ on $X_{M_2}$.  That is,
\begin{equation}
\label{Eq:projequiv}
X_{M_2}\Theta=\mathcal{P}_{M_2}X_{M_1}
\end{equation}
Note that on $\mathcal{B}$,
\begin{equation*}
X\beta=\tilde{X}\hat{\beta}_{M^{(-)}}+\tau X_{M_1}\hat{\beta}_{M_1}+(1-\tau)X_{M_2}\Theta\hat{\beta}_{M_1}
\end{equation*}
Recalling that $\tilde{Y}=Y-\tilde{X}\hat{\beta}_{M^{(-)}}$, minimizing \eqref{Eq:ENDef} on $\mathcal{B}$ is equivalent to minimize
\begin{align}
\label{Eq:whattomingroups}
 \nonumber
||\tilde{Y}-\tau X_{M_1}\hat{\beta}_{M_1}-(1-\tau)X_{M_2}\Theta\hat{\beta}_{M_1}||_2^2+&\lambda_1[\tau|| \hat{\beta}_{M_1}||_1+(1-\tau)||\Theta\hat{\beta}_{M_1}||_1] \\
+&\lambda_2[\tau^2||\hat{\beta}_{M_1}||_2^2+(1-\tau)^2||\Theta\hat{\beta}_{M_1}||_2^2]
\end{align}
as a function of $\tau$. Using a first-order condition and substituting \eqref{Eq:projequiv} we find that \eqref{Eq:whattomingroups} is minimized for
\begin{align}
\label{Eq:taustar}
\nonumber
&\tau^*=\\
&\frac{-(\tilde{Y}-\mathcal{P}_{M_2}X_{M_1}\hat{\beta}_{M_1})^T(I-\mathcal{P}_{M_2})X_{M_1}\hat{\beta}_{M_1}+\frac{\lambda_1}{2}(||\hat{\beta}_{M_1}||_1-||\Theta\hat{\beta}_{M_1}||_1)-\lambda_2||\Theta\hat{\beta}_{M_1}||_2^2}
{||(I-\mathcal{P}_{M_2})X_{M_1}\hat{\beta}_{M_1}||_2^2-\lambda_2||\hat{\beta}_{M_1}||_2^2-\lambda_2||\Theta\hat{\beta}_{M_1}||_2^2}.
\end{align}
Before we continue, note that if $X_2=X_1$ then $\Theta$ is the identity matrix and $\mathcal{P}_{M_2}X_{M_1}=X_1$. Substituting these facts into \eqref{Eq:taustar}, we get that $\tau^*=\frac{1}{2}$ as one might expect. Same result is obtained for the case $M_2\subseteq M_1$.

As it can be seen in \eqref{Eq:subspaceBgroups}, the coordinates of $\hat{\beta}_{M_2}$ are all different than zero if $\tau^*<1$. Now, since $\mathcal{P}_{M_2}(I-\mathcal{P}_{M_2})=0$ we get that  $\tau^*<1$ if
\begin{equation*}
\begin{split}
-\tilde{Y}^T(I-\mathcal{P}_{M_2})X_{M_1}\hat{\beta}_{M_1}+||(I-\mathcal{P}_{M_2})X_{M_1}\hat{\beta}_{M_1}&||_2^2-\frac{\lambda_1}{2}||\Theta\hat{\beta}_{M_1}||_1\\
&>-\frac{\lambda_1}{2}||\hat{\beta}_{M_1}||_1-\lambda_2||\hat{\beta}_{M_1}||_2^2
\end{split}
\end{equation*}
which is certainly true if
\begin{equation}
\label{Eq:almostend}
\tilde{Y}^T\mathcal{P}_{M_2}X_{M_1}\hat{\beta}_{M_1}-\frac{\lambda_1}{2}||\Theta\hat{\beta}_{M_1}||_1>-\frac{\lambda_1}{2}||\hat{\beta}_{M_1}||_1-\lambda_2||\hat{\beta}_{M_1}||_2^2+\tilde{Y}^TX_{M_1}\hat{\beta}_{M_1}
\end{equation}
which is true if the condition in \eqref{Eq:CondENgroups} is fulfilled for the appropriate $c_1$.
$\hfill \square$
 \section{Supplementary tables for Section \ref{SubSec:data}}
\begin{longtable}{|c|rrrrrrrrrr|}
   \caption{The 112 models selected for the riboflavin data. Each row is a model, the numbers are the column number (the gene) in $X$.}
   \label{Tab:RiboModels}\\
  \hline
Model &  &  &  &  &  &  &  &  &  &  \\
  \hline
1 & 1278 &  &  &  &  &  &  &  &  &  \\
  2 & 1279 &  &  &  &  &  &  &  &  &  \\
  3 & 4003 &  &  &  &  &  &  &  &  &  \\
  4 & 1516 &  &  &  &  &  &  &  &  &  \\
  5 & 1312 &  &  &  &  &  &  &  &  &  \\
  6 & 1278 & 4003 &  &  &  &  &  &  &  &  \\
  7 & 1303 & 4003 &  &  &  &  &  &  &  &  \\
  8 & 1279 & 4003 &  &  &  &  &  &  &  &  \\
  9 & 1278 & 4006 &  &  &  &  &  &  &  &  \\
  10 & 1279 & 4004 &  &  &  &  &  &  &  &  \\
  11 & 69 & 2564 & 4003 &  &  &  &  &  &  &  \\
  12 & 73 & 2564 & 4003 & & & & & & & \\
  13 & 144 & 2564 & 4003 & & & & & & &  \\
  14 & 69 & 2564 & 4004 & & & & & & & \\
  15 & 69 & 2564 & 4006 & & & & & & & \\
  16 & 792 & 1478 & 4002 & & & & & & & \\
  17 & 792 & 1478 & 4003 & & & & & & & \\
  18 & 792 & 1478 & 4004 & & & & & & & \\
  19 & 73 & 1279 & 2564 & 4004 & & & & & & \\
  20 & 73 & 1279 & 2564 & 4003 & & & & & & \\
  21 & 144 & 1279 & 2564 & 4004 & & & & & & \\
  22 & 73 & 1849 & 2564 & 4004 & & & & & & \\
  23 & 73 & 1279 & 2564 & 4006 & & & & & & \\
  24 & 144 & 1279 & 2564 & 4003 & & & & & & \\
  25 & 73 & 1279 & 1849 & 2564 & 4003 & & & & & \\
  26 & 69 & 1849 & 2564 & 3226 & 4003 & & & & & \\
  27 & 69 & 1425 & 1640 & 2564 & 4003 & & & & & \\
  28 & 69 & 1849 & 2564 & 3226 & 4004 & & & & & \\
  29 & 69 & 1425 & 1640 & 2564 & 4004 & & & & & \\
  30 & 144 & 792 & 1849 & 2564 & 4003 & & & & & \\
  31 & 144 & 1849 & 2564 & 3226 & 4003 & & & & & \\
  32 & 73 & 974 & 1279 & 2564 & 4003 & & & & & \\
  33 & 144 & 1278 & 1425 & 2564 & 4003 & & & & & \\
  34 & 73 & 792 & 2116 & 2564 & 4004 & & & & & \\
  35 & 73 & 1278 & 1849 & 2564 & 4004 & & & & & \\
  36 & 144 & 1278 & 1849 & 2564 & 4003 & & & & & \\
  37 & 73 & 1279 & 1425 & 2564 & 4004 & & & & & \\
  38 & 144 & 792 & 1849 & 2027 & 2564 & 4004 & & & & \\
  39 & 73 & 974 & 1278 & 1849 & 2564 & 4003 & & & & \\
  40 & 69 & 415 & 1849 & 2564 & 3226 & 4004 & & & & \\
  41 & 73 & 792 & 974 & 2116 & 2564 & 4003 & & & & \\
  42 & 73 & 1279 & 1640 & 1849 & 2564 & 4004 & & & & \\
  43 & 69 & 315 & 792 & 1849 & 2564 & 4004 & & & & \\
  44 & 73 & 792 & 1849 & 2027 & 2564 & 4004 & & & & \\
  45 & 73 & 792 & 1303 & 2116 & 2564 & 4003 & & & & \\
  46 & 69 & 792 & 1849 & 2027 & 2564 & 4003 & & & & \\
  47 & 69 & 792 & 1282 & 1849 & 2564 & 4003 & & & & \\
  48 & 69 & 792 & 1131 & 1849 & 2564 & 4003 & & & & \\
  49 & 73 & 1131 & 1278 & 1524 & 2564 & 4006 & & & & \\
  50 & 144 & 1131 & 1303 & 1524 & 2564 & 4006 & & & & \\
  51 & 73 & 792 & 1528 & 1849 & 2564 & 4003 & & & & \\
  52 & 73 & 792 & 1294 & 2116 & 2564 & 4003 & & & & \\
  53 & 69 & 1131 & 1278 & 1524 & 1762 & 2564 & 4006 & & & \\
  54 & 144 & 1279 & 1762 & 1820 & 2027 & 2564 & 4004 & & & \\
  55 & 69 & 1279 & 1425 & 1640 & 1820 & 2564 & 4006 & & & \\
  56 & 144 & 792 & 1312 & 1849 & 2027 & 2564 & 4004 & & & \\
  57 & 73 & 1278 & 1762 & 1820 & 1857 & 2564 & 4003 & & & \\
  58 & 73 & 1131 & 1279 & 1524 & 1528 & 2564 & 4003 & & & \\
  59 & 69 & 792 & 1303 & 1849 & 2484 & 2564 & 4003 & & & \\
  60 & 69 & 792 & 1639 & 1849 & 2027 & 2564 & 4003 & & & \\
  61 & 73 & 315 & 792 & 1278 & 1524 & 2564 & 4004 & & & \\
  62 & 69 & 315 & 1425 & 1524 & 1640 & 2564 & 4004 & & & \\
  63 & 73 & 1131 & 1279 & 1857 & 2116 & 2564 & 4004 & & & \\
  64 & 144 & 1131 & 1279 & 1857 & 2116 & 2564 & 4004 & & & \\
  65 & 144 & 1101 & 1131 & 1279 & 1762 & 2564 & 4004 & & & \\
  66 & 69 & 792 & 1131 & 1849 & 2564 & 3514 & 4004 & & & \\
  67 & 73 & 792 & 1279 & 1478 & 2027 & 2564 & 4002 & & & \\
  68 & 73 & 792 & 1131 & 1279 & 1312 & 2116 & 2564 & 4004 & & \\
  69 & 73 & 315 & 792 & 1279 & 1312 & 2116 & 2564 & 4004 & & \\
  70 & 73 & 315 & 792 & 1279 & 1503 & 2116 & 2564 & 4004 & & \\
  71 & 69 & 792 & 1131 & 1279 & 2116 & 2564 & 3288 & 4006 & & \\
  72 & 73 & 315 & 1279 & 1762 & 1849 & 2564 & 3288 & 4004 & & \\
  73 & 144 & 974 & 1131 & 1279 & 1524 & 2564 & 3514 & 4003 & & \\
  74 & 144 & 974 & 1131 & 1279 & 1425 & 1524 & 2564 & 4003 & & \\
  75 & 69 & 792 & 859 & 1131 & 1279 & 2116 & 2564 & 4004 & & \\
  76 & 73 & 792 & 1279 & 1849 & 2484 & 2564 & 4004 & 4006 & & \\
  77 & 73 & 974 & 1101 & 1131 & 1279 & 2564 & 3105 & 4004 & & \\
  78 & 73 & 792 & 1131 & 1312 & 2116 & 2242 & 2564 & 4004 & & \\
  79 & 73 & 792 & 974 & 1131 & 1364 & 2116 & 2564 & 4004 & & \\
  80 & 73 & 792 & 1131 & 1312 & 1639 & 2116 & 2564 & 4004 & & \\
  81 & 73 & 792 & 1131 & 1364 & 2116 & 2242 & 2564 & 4004 & & \\
  82 & 73 & 792 & 1131 & 1312 & 2116 & 2564 & 3905 & 4004 & & \\
  83 & 144 & 1131 & 1279 & 1524 & 1528 & 2484 & 2564 & 3465 & 4004 &  \\
  84 & 73 & 974 & 1131 & 1279 & 1524 & 2242 & 2564 & 3514 & 4006 &  \\
  85 & 73 & 974 & 1131 & 1279 & 1524 & 2242 & 2564 & 3465 & 4006 &  \\
  86 & 73 & 244 & 792 & 1131 & 1278 & 2116 & 2564 & 3104 & 4006 &  \\
  87 & 73 & 792 & 1131 & 1278 & 1297 & 2116 & 2564 & 3104 & 4006 &  \\
  88 & 144 & 1101 & 1279 & 1425 & 1640 & 2116 & 2484 & 2564 & 4004 &  \\
  89 & 69 & 859 & 1101 & 1640 & 1762 & 2484 & 2564 & 3226 & 4003 &  \\
  90 & 73 & 315 & 792 & 1303 & 1849 & 2564 & 4004 & 4006 & 4045 &  \\
  91 & 73 & 144 & 315 & 792 & 1279 & 1849 & 2462 & 2564 & 4004 &  \\
  92 & 73 & 315 & 792 & 1303 & 1849 & 2564 & 4004 & 4045 & 4075 &  \\
  93 & 73 & 827 & 1131 & 1279 & 1639 & 2242 & 2564 & 3465 & 4003 &  \\
  94 & 69 & 1312 & 1425 & 1640 & 1762 & 2116 & 2564 & 3104 & 4004 &  \\
  95 & 69 & 1312 & 1425 & 1528 & 1640 & 1762 & 2116 & 2564 & 4004 &  \\
  96 & 73 & 624 & 792 & 1131 & 1278 & 1849 & 1855 & 2564 & 4004 &  \\
  97 & 144 & 827 & 1131 & 1279 & 1639 & 2242 & 2564 & 3465 & 4003 &  \\
  98 & 73 & 792 & 1131 & 1279 & 2027 & 2116 & 2564 & 3104 & 3288 & 4004 \\
  99 & 73 & 974 & 1131 & 1279 & 1364 & 1524 & 2027 & 2564 & 3104 & 4003 \\
  100 & 73 & 859 & 974 & 1131 & 1279 & 1364 & 1524 & 2484 & 2564 & 4003 \\
  101 & 73 & 624 & 792 & 1131 & 1279 & 1849 & 2116 & 2564 & 3226 & 4004 \\
  102 & 73 & 1303 & 1524 & 1762 & 2027 & 2484 & 2564 & 3905 & 4004 & 4075 \\
  103 & 69 & 974 & 1425 & 1524 & 1640 & 2116 & 2484 & 2564 & 3288 & 4004 \\
  104 & 69 & 974 & 1278 & 1425 & 1524 & 1640 & 2484 & 2564 & 3288 & 4004 \\
  105 & 73 & 315 & 1279 & 1294 & 1762 & 2027 & 2462 & 2564 & 4003 & 4075 \\
  106 & 69 & 974 & 1278 & 1425 & 1524 & 1640 & 2484 & 2564 & 3105 & 4004 \\
  107 & 69 & 974 & 1425 & 1524 & 1640 & 1857 & 2484 & 2564 & 3288 & 4004 \\
  108 & 73 & 859 & 1131 & 1278 & 1297 & 1524 & 2462 & 2484 & 2564 & 4003 \\
  109 & 73 & 1278 & 1425 & 1524 & 1639 & 2484 & 2564 & 3465 & 4003 & 4045 \\
  110 & 73 & 792 & 1131 & 1279 & 1478 & 2116 & 2484 & 2564 & 3465 & 4006 \\
  111 & 73 & 859 & 1131 & 1278 & 1503 & 1524 & 2462 & 2484 & 2564 & 4003 \\
  112 & 73 & 415 & 859 & 1131 & 1278 & 1503 & 1524 & 2484 & 2564 & 4003 \\
   \hline
\end{longtable}
\begin{table}[ht!]
\begin{center}
\caption{Potential predictors for mortality rate in Section \ref{SubSubSec:mortData}}
\label{Tab:mortPredictors}
\begin{tabular}{ll}
\hline
Predictor & Description\\
\hline
\textbf{prec} & Mean annual precipitation in inches \\
\textbf{jant} & Mean January temperature in degrees F\\
\textbf{jult} & Mean July temperature in degrees F\\
\textbf{age65} & Percentage of population aged 65 or older\\
\textbf{pphs}& Population per household\\
\textbf{educ}& Median school years completed by those over 22\\
\textbf{facl}& Percentage of housing units which are sound and with all facilities\\
\textbf{dens}& Population per square mile in urbanized areas\\
\textbf{nwht}& Percentage of non-white population in urbanized areas\\
\textbf{wtcl}& Percentage of employed in white collar occupations\\
\textbf{linc}& Percentage of families with income $<$ 3,000 dollars in urbanized\\ & areas\\
\textbf{HC}& Relative pollution potential of hydrocarbon\\
\textbf{NOx}& Relative pollution potential of nitric oxides\\
\textbf{SUL}& Relative pollution potential of sulphur dioxide\\
\textbf{hum}& Annual average percentage of relative humidity at 1pm\\
\hline
\end{tabular}
\end{center}
\end{table}
\pagebreak
\bibliography{highdim}{}
\bibliographystyle{newapa}
\end{document}